\documentclass[11pt]{amsart}
\usepackage{amsmath}
\usepackage{amssymb,latexsym}
\usepackage{color,enumitem}
\usepackage{hyperref}
\usepackage[capitalize]{cleveref}
\usepackage[utf8]{inputenc}
\usepackage[nospace,noadjust]{cite}
\usepackage{relsize}
\usepackage{cite}
\usepackage{color}
\usepackage[dvipsnames]{xcolor}
\usepackage{hyperref, enumitem}
\hypersetup{
  colorlinks   = true, 
  urlcolor     = blue, 
  linkcolor    = Purple, 
  citecolor   = red 
}
\theoremstyle{definition}
\newtheorem{theorem}{Theorem}[section]

\newtheorem{lemma}[theorem]{Lemma}
\newtheorem{corollary}[theorem]{Corollary}
\newtheorem{definition}[theorem]{Definition}
\newtheorem{proposition}[theorem]{Proposition}

\newtheorem{example}[theorem]{Example}

\newtheorem{remark}[theorem]{Remark}

\newtheorem{notation}[theorem]{Notation}

\newcommand{\fqn}{\mathbb{F}_{q^n}}
\newcommand{\fq}{\mathbb{F}_{q}}
\newcommand{\fqk}{\mathbb{F}_{q^k}}

\newcommand{\cA}{{\mathcal A}}

\newcommand{\cC}{{\mathcal C}}

\newcommand{\cG}{{\mathcal G}}

\newcommand{\cF}{{\mathcal F}}

\newcommand{\cU}{{\mathcal U}}
\newcommand{\cL}{{\mathcal L}}

\newcommand{\NN}{{\mathbb N}}
\newcommand{\F}{{\mathbb F}}
\newcommand{\Tr}{\mathrm{Tr}}
\newcommand{\N}{\mathrm{N}}
\newcommand{\lmb}{\lambda}
\newcommand{\la}{\langle}
\newcommand{\ra}{\rangle}
\newcommand{\PG}{\mathrm{PG}}
\newcommand{\subspace}[1]{\mbox{$\langle{#1}\rangle$}}
\newcommand{\Aut}{\mbox{\rm Aut}}
\newcommand{\Fqn}{\mbox{$\F_{q^n}$}}

\newcommand{\rrk}{\mbox{-{\rm Rank}}}
\newcommand{\GL}{{\rm GL}}
\newcommand{\id}{{\rm id}}
\newcommand{\mmod}{\mbox{\rm mod }}
\newcommand{\lcm}{{\rm lcm}}
\newcommand{\Gal}{\mbox{$\text{Gal}$}}

\newcommand{\Orb}{\mbox{$\text{Orb}$}}
\newcommand{\Smallfourmat}[4]{\mbox{$\left(\begin{smallmatrix}{#1}&{#2}\\[.4ex]{#3}&{#4}\end{smallmatrix}\right)$}}

\renewcommand{\P}{{\mathbb P}}
\newcommand{\Z}{{\mathbb Z}}

\newcommand{\cI}{{\mathcal I}}
\newcommand{\cO}{{\mathcal O}}

\begin{document}

\title{Quasi-optimal cyclic orbit codes}
\author{Chiara Castello, Heide Gluesing-Luerssen, Olga Polverino, and Ferdinando Zullo}
\date{December 31, 2024}

\maketitle

\begin{abstract}
We focus on two aspects of cyclic orbit codes: 
invariants under equivalence and quasi-optimality.
Regarding the first aspect, we establish a connection between the codewords of a cyclic orbit code and a certain linear set on the projective line.
This allows us to derive new  bounds on the parameters of the code. In the second part, we study a particular family of (quasi-)optimal cyclic orbit codes  and derive a general existence theorem for quasi-optimal codes in even-dimensional vector spaces over finite fields of any characteristic. 
Finally, for our particular code family we describe the automorphism groups under the general linear group and a suitable Galois group.
\end{abstract}

\section{Introduction}

A network channel can be represented as a directed acyclic graph with three types of vertices: sources (vertices without incoming edges), sinks (vertices without outgoing edges) and inner nodes (vertices with both incoming and outgoing edges).
When inner nodes combine incoming information linearly, this is named \textbf{linear network coding}. If the network's topology is unknown in advance, it is called \textbf{random network coding}.
Network coding was developed in the seminal paper~\cite{ACLY00}
to enhance throughput in multicast sessions. Unlike traditional routing, where intermediate nodes merely store and forward data, network coding enables these nodes to actively combine and encode the data before transmission.
This paradigm shift has driven interest in network coding, with applications in areas such as wireless communications and cloud computing; see e.g. \cite{Netsurvey}.
In \cite{KoetterK}, the authors laid the mathematical foundation for random linear network coding, highlighting the effectiveness of subspace codes for reliable transmissions. These codes leverage the model's ability to transmit linear combinations of vectors while preserving their span. More precisely, codewords are subspaces of $\mathbb{F}^n$, for some field $\mathbb{F}$, and a code is a collection of subspaces equipped with the subspace distance.
Since then, intensive research has focused on the theory of subspace codes, aiming to establish bounds, develop constructions, and design efficient decoding algorithms; see \cite{Subcodessurvey}.

In this paper we focus on cyclic orbit codes, introduced in \cite{Etzion}, as the subspace code analogue of classical linear codes in the Hamming metric.
The set of all $k$-dimensional $\F_q$-subspaces of $\F_{q^n}$ is called the \textbf{Grassmannian} over $\F_q$ and denoted by $\cG_q(k,n)$. 
A \textbf{constant dimension subspace code} is a subset $\mathcal{C}$ of $\cG_q(k,n)$ endowed with the \textbf{subspace distance} 
\[d(U,V)=2(k-\dim_{\F_q}(U \cap V))\]
for $U,V \in\cG_q(k,n)$. 
As usual, the \textbf{minimum distance} of~$\cC$ is defined as $d(\cC)=\min\{d(U,V)\colon\!\!$ $U,V\in\cC,\,U\neq V\}$.
Clearly, it is upper bounded by~$2k$.
A subspace code $\mathcal{C} \subseteq \cG_q(k,n)$ is said to be \textbf{cyclic} if $\alpha V \in \mathcal{C}$
for all $\alpha \in \F_{q^n}^*$ and $V \in \mathcal{C}$. 
For $U \in\cG_q(k,n)$ the \textbf{orbit} of $U$ is the set $\mathrm{Orb}(U)=\{\alpha U: \alpha \in \F_{q^n}^*\}$.
Clearly, $\mathrm{Orb}(U)$ is a cyclic subspace code, called a \textbf{cyclic orbit code} (or a single-orbit or one-orbit code \cite{bhaintwal2024,CPZ2024IEEE,CPZ2024FullWeight}).
The size of a cyclic orbit code is upper bounded by $(q^n-1)/(q-1)$ and codes with this size are called \textbf{full-length orbit codes}.

It is not hard to see that for $k>2$ a full-length orbit code cannot attain distance~$2k$ (see also \cite{manganiello2008spread} and \cite{Segre}).
For this reason we call a full-length orbit code \textbf{optimal} if its distance is $2(k-1)$, and \textbf{quasi-optimal} if it is $2(k-2)$.
In \cite{Trautmann}, the authors conjectured the existence of optimal full-length orbit codes 
for all positive integers $n$ and $2<k\leq n/2$ and any prime power $q$.
A lot of progress has been made on establishing this conjecture, see e.g.~\cite{BEGR, Otal, SantZullo1, Roth}.
However, the case $q=2$ and $k=n/2$ remains largely open and in particular, the conjecture is not true for $(n,k,q)=(8,4,2)$; see \cite{GLTroha15}.
For further details we refer to Remarks~\ref{R-n2} and~\ref{rmk:UsgammaoptimaliffNormneq1}.
On the other hand, not much is known about quasi-optimal full-length orbit codes.

In order to achieve a finer distinction between cyclic orbit codes than the minimum distance, the notion of distance distribution has been introduced in \cite{heideweight}.
Therein, it has been shown that the distance distribution of an optimal full-length orbit code only depends on the parameters $(n,k,q)$, whereas this is not the case for quasi-optimal full-length orbit codes. 
In the latter case,  the distance distribution of $\Orb(U)$ depends on yet another parameter, namely the number of fractions $\overline{uv^{-1}}\in\P(\fqn)$ for $u,v \in U\setminus\{0\}$, denoted by~$f_U$.

In this paper, we establish a relation between~$f_U$ and the size of the $\fq$-linear set $L_{U\times U}$, the set of points in $\mathrm{PG}(1,q^n)$ defined by vectors in $U\times U$. 
With the aid of known  lower bounds on the size of a linear set (see \cite{DeBeuleVdV} and \cite{csajbok2023maximum}) we are able to improve upon previous bounds on~$f_U$ provided in \cite{heideweight} (under some mild constraints on the distance distribution of the code); see \cref{cor:w2k-2neq0thenfU} and   \cref{th:bounds f_U}.
This allows us to narrow down the possible distance distributions of quasi-optimal full-length orbit codes; see  \cref{cor:fUconsequencequasi-optimal}.
All of this provides answers to some questions posed in \cite[Section 6]{heideweight}.
In particular, for $k=3$, exactly three different distance distributions arise.

Another contribution of this paper is a detailed study of the orbit codes $\cO_{s,\gamma}=\Orb(U_{s,\gamma})$, where $U_{s,\gamma}=\lbrace u+u^{q^s}\gamma\colon u\in\fqk\rbrace\in\cG_q(k,n)$, where $n=2k,\,s\in[k]$ such that $\gcd(s,k)=1$ and $\gamma\in\fqn$ such that $\fqn=\fqk(\gamma)$.
These codes have been introduced in \cite{Roth}, and the authors proved that if $\gamma^{(q^n-1)/(q-1)}\ne 1$ then $\cO_{s,\gamma}$ is an optimal 
full-length orbit code.
Based on this, we show that $\cO_{s,\gamma}$ is always full-length and either optimal or quasi-optimal, and the distinction depends only on the field norm 
of~$\gamma$; see  \cref{P-OptQuasiOpt}.
It is not hard to determine the number of optimal and quasi-optimal orbit codes $\cO_{s,\gamma}$; see \cref{P-CountOpt}.
All of this allows us to establish in \cref{thm:existencequasi-optimal} the existence of quasi-optimal full-length orbit codes for any $q\geq2,\,n$ even, and 
$3\leq k\leq n/2$.

As shown in \cite{heideweight}, if~$n$ is even, then the distance distribution of a quasi-optimal full-length orbit code $\Orb(U)$ depends on whether or 
not~$U$ contains a cyclic shift of the subfield $\F_{q^2}$ (thus the same is true for the parameter~$f_U$). 
In \cref{T-ContainCycShift} we characterize the subspaces $U_{s,\gamma}$ containing such a cyclic shift and in \cref{P-CycShift} we determine the
number of the resulting orbit codes $\cO_{s,\gamma}$.

Finally, we study the codes $\cO_{s,\gamma}$ up to equivalence. 
The study of equivalence for subspace codes has been initiated in \cite{Trautmann2}, and in \cite{Heideequiv} special attention has been paid to the 
case of cyclic orbit codes.
We consider two types of equivalence: $\F_{p^t}$-linear isometries and 
$\F_{p^t}$-Frobenius isometries (where $q=p^h$ and $t$ a divisor of $h$); see \cref{D-Isometric} for details. 
In the former case the codes are related via an $\F_{p^t}$-isomorphism on~$\F_{q^n}$, while in the latter they are related via an isomorphism from the 
Galois group $\Gal(\fqn\mid\F_{p^t})$.
With the aid of a result in \cite{Heideequiv}, we prove in \cref{T-LinFrobIso} that if $U_1,\,U_2\in\cG_q(k,n)$ are generic subspaces, then $\Orb(U_1)$ and  
$\Orb(U_2)$ are $\F_{p^t}$-linearly isometric if and only if they are $\F_{p^t}$-Frobenius isometric. 
This allows us to use \cite{CPSZSidon} for determining the $\F_{p^t}$-auto\-morphism group of the codes $\cO_{s,\gamma}$; see \cref{P-AutGroup}. 
Finally, in \cref{T-AutGroupellhat} we provide, for the codes $\cO_{s,\gamma}$, a description of the $\F_p$-Frobenius-automorphism groups
and thus of the sizes of the orbits of these codes under the action of the Galois group.
Examples illustrate our findings.

\section{Linear sets}\label{S-LinSets}

In this paper, we will make use of $\fq$-linear sets in the projective line and therefore start with recalling some useful definitions and results regarding linear sets in $\PG(1,q^n)$. We refer to \cite{LavVdV} and \cite{Polverino} for comprehensive references on linear sets in projective spaces of any finite dimension.
Let $V$ be a $2$-dimensional $\fqn$-vector space and let $\Lambda=\PG(V,\fqn)=\PG(1,q^n)$. 
Let $W \neq \{\mathbf{0}\}$ be an $\fq$-subspace of $V$. The pair $(W,L_W)$ where
\[ L_W=\{\langle \mathbf{u} \rangle_{\fqn} \colon \mathbf{u}\in W\setminus\{\mathbf{0}\}\}, \]
is said to be an $\fq$-\textbf{linear set} of $W$-rank $\dim_{\fq}(W)$.
Note that the set $L_W \subseteq \Lambda$ does not uniquely determine the dimension of the subspace~$W$: we may have $L_W=L_U$ 
and $\dim_{\fq}(W) \neq \dim_{\fq}(U)$. We denote the $W$-rank of the linear set $L_W$ by $W\rrk(L_W)$.

Another important notion is the following. Consider an $\fq$-linear set $L_W$ of $W$-rank $k$ in $\PG(1,q^n)$. 
The $W$-\textbf{weight} of a point $P=\langle \mathbf{v}\rangle_{\fqn} \in \Lambda$ is defined as 
$w_{W}(P) :=\dim_{\fq}(W \cap \langle \mathbf{v}\rangle_{\fqn})$ (of course, this dimension depends on the choice of the 
ground field~$\fq$; we will omit this dependence in our notation).
For $i\in \{0,\ldots,k\}$ denote by~$N_i$ the number of points of $\Lambda$ having $W$-weight~$i$.
We call $(N_0,\ldots,N_k)$ the $W$-\textbf{weight distribution} of $L_W$.
Note that if $\mathrm{Rank}(L_W)> n$, then $N_0=0$ and hence $L_W=\Lambda$.
We say that~$L_W$ has $W$-\textbf{weight spectrum}  $(i_1,\ldots,i_t)$ 
with $1\leq i_1 <i_2 < \ldots <i_t\leq k$ if the $W$-weight of every point $P \in L_W$ is in $\{i_1,\ldots,i_t\} $ 
and each of these integers $i_j$ occurs as the weight of at least one point of $L_W$.

Naturally, the size of the linear set is given by
\[
    |L_W| =N_1+\ldots+N_k,
\]
and $N_0 =q^n+1-|L_W|$.
Counting the pairs $(P,\mathbf{w})$, where $P$ is a point of $L_W$ and $\mathbf{w}\in W$ such that $P=\la \mathbf{w}\ra_{\fqn}$, we obtain
\begin{equation}\label{eq:pesi}
    N_1(q-1)+N_2(q^2-1)+\ldots+N_k(q^k-1)=q^k-1,
\end{equation}
and thus $N_1+N_2(q+1)+\ldots+N_k(q^{k-1}+\ldots+q+1)=q^{k-1}+\ldots+q+1$.
From the previous equalities we get
\begin{equation}\label{eq:card}
    |L_W| \leq \frac{q^k-1}{q-1}.
\end{equation}
Moreover, by the Grassmann formula,  we also get
\[
    w_{W}(P)+w_{W}(Q)\leq W\rrk(L_W),
\]
for any $P,Q \in \mathrm{PG}(1,q^n)$ with $P\ne Q$. If there exist two distinct points $P,Q \in L_W$ such that $ w_{W}(P)+w_{W}(Q)= W\rrk(L_W),$ then we say that $L_W$ is a linear set with points of {\bf complementary weights}. 
The following result has been proved in \cite[Theorem 1.2]{DeBeuleVdV} (and for $k=n$ in \cite[Lemma 2.2]{BoPol}). 

\begin{theorem}\label{th:sizelinset}
An $\fq$-linear set $L_W$ of $W$-rank $k\leq n $ in $\PG(1,q^n)$ with at least one point of $W$-weight one has at least $q^{k-1}+1$ points.
\end{theorem}

\begin{remark}
\label{rmk:sizelinsetk=n+1}
Note that Theorem \ref{th:sizelinset} also holds when $k=n+1$. 
Indeed, in that case we have $L_W=\PG (1,q^n)$ and so $|L_W|=q^n+1=q^{k-1}+1$.
\end{remark}

In this paper, we will mainly consider $V=\fqn\times\fqn$ and linear sets of the form
\[ L_{U\times U}=\{\langle (u_1,u_2)\rangle_{\fqn} \colon u_1,u_2 \in U, \,\,(u_1,u_2)\ne(0,0)\} \]
for some $k$-dimensional $\fq$-subspace $U$ of $\fqn$. 
Thus $(U\!\!\times\!\! U)\rrk(L_{U \times U})=2k$.
Note that the projective space is given by 
$\Lambda=\PG(V,q^n)=\{\subspace{(1,\alpha)}_{\fqn}\colon \alpha\in\fqn\}\cup\{\subspace{(0,1)}_{\fqn}\}$.
Now we compute
\begin{equation}\label{e-Weight01}
  w_{U\times U}(\subspace{(1,0)}_{\fqn})=\dim_{\fq}(U\times \{0\})=\dim_{\fq} (U)=k= w_{U\times U}(\subspace{(0,1)}_{\fqn}).
\end{equation}
Furthermore, for any $\alpha\in\fqn^*$ we have  
\begin{equation}\label{e-UUalpha}
    (U\times U)\cap \subspace{(1,\alpha)}_{\fqn}=\{(u,\alpha u)\colon u\in U\cap\alpha^{-1} U\}
\end{equation}
and thus
\begin{equation}\label{e-lessk}
   w_{U\times U}(\subspace{(1,\alpha)}_{\fqn})=\dim_{\fq}(U\cap\alpha^{-1}U)=\dim_{\fq}(U\cap\alpha U)\leq k.
\end{equation}
In particular, $L_{U \times U}$  is a linear set with two points of complementary weights.
Moreover, if $\alpha$ is in~$\fq$, then $w_{U\times U}(\la (1,\alpha)\ra_{\fqn})=k$, and therefore
\[
    N_k\geq q+1.
\]
Furthermore, we have $N_k = q+1$ iff $\alpha U\neq U$ for all $\alpha\in\fqn\setminus\fq$.
This last condition means that~$U$ is not closed under multiplication by any scalars outside~$\fq$. 
We cast the following definition.

\begin{definition}\label{D-MaxFieldLin}
Let $W$ be an $\F_q$-subspace of~$\fqn$.
We call $\F_{q^t}$ the \textbf{maximum field of linearity} of~$W$ if it is the largest subfield of $\fqn$ such that~$W$ is 
closed under multiplication by scalars from~$\F_{q^t}$. 
\end{definition}

We have the following upper bound for the size of $L_{U\times U}$, which improves upon \eqref{eq:card}.

\begin{proposition} \label{upper boundlinearsetcompwe}
Let $U$ be an $\fq$-subspace of $\fqn$ of dimension $k\leq n$ and assume that $\fq$ is the maximum field of linearity of $U$.  
Let $e=\min\{w_{U\times U}(P) \colon P \in L_{U\times U} \}$.
Then
\[
    |L_{U\times U}|\leq \frac{(q^k-1)(q^k+1)}{q^e-1}-(q+1)\frac{(q^k-q^e)}{q^e-1}. 
\]
\end{proposition}

\begin{proof}
From \eqref{e-lessk}  we know already that the maximum weight of the points of $L_{U\times U}$ is $k$. 
Furthermore, for any $\alpha\in\fqn^*$ we have 
$w_{U\times U}(\langle(1,\alpha)\rangle_{\fqn})=k\Longleftrightarrow\dim_{\fq}(U\cap\alpha U)=k\Longleftrightarrow\alpha\in\fq$,
where the last step follows from the fact that $\fq$ is the maximum field of linearity of $U$.
Together with \eqref{e-Weight01} we arrive at $N_k=q+1$.
Since $N_1=N_2=\ldots=N_{e-1}=0$, we obtain from  \eqref{eq:pesi} 
\[
   N_e(q^e-1)+\ldots+N_{k-1}(q^{k-1}-1)+(q+1)(q^k-1)=q^{2k}-1.
 \]
Using the previous equality together with $q^i-1\geq q^e-1$ for any $i\in\{e,\ldots, k-1\}$ and the identity $|L_{U\times U}|=N_e+\ldots+N_{k-1}+q+1$, 
we get 
\[
   (|L_{U\times U}|-(q+1))(q^e-1)+(q+1)(q^k-1)\leq q^{2k}-1 
\]
and therefore
\[
    |L_{U\times U}|\leq \frac{(q^k-1)(q^k+1)}{q^e-1}-(q+1)\frac{(q^k-q^e)}{q^e-1}. 
    \qedhere
\]
\end{proof}

\section{Cyclic orbit codes}

In this short section we revisit some basic notions about cyclic orbit codes.
Recall the notation $\cG_q(k,n)$ for the Grassmannian of $k$-dimensional $\F_q$-subspaces of~$\fqn$.

\begin{theorem}\cite[Def.~19]{Trautmann}, \cite[Def.~2.1]{GLTroha15}\label{th:orbitsize}
Let $U\in\cG_q(k,n)$. The \textbf{cyclic orbit code generated by}~$U$ is defined as 
$\Orb(U)=\{\alpha U\colon \alpha\in\fqn^*\}$, and $\mathrm{Stab}(U)=\{ \alpha \in \fqn^* \colon \alpha U=U \}$ is called 
the \textbf{stabilizer} of $U$.
It is thus the multiplicative group of the maximum field of linearity; see \cref{D-MaxFieldLin}.
Clearly, a subfield $\F_{q^t}$ is the maximum field of linearity of~$U$ if and only if 
$|\Orb(U)|=\frac{q^n-1}{q^t-1}$.
In the case where $t=1$, we call $\Orb(U)$ a \textbf{full-length orbit code}.
\end{theorem}

Note that the minimum distance of $\Orb(U)$ reads as 
\[ 
   d(\Orb(U))=2k-2\max\{ \dim_{\fq}(U\cap \alpha U) \colon \alpha \in \fqn^*\setminus\mathrm{Stab}(U)\}. 
\]

For the next definition, we need the following notation.
Let $\F_{q^t}$ be a subfield of~$\fqn$.
The projective space over $\fqn$ considered as an $\F_{q^t}$-vector space is denoted by $\mathbb{P}_t(\fqn)$.
Its points are of the form $\alpha\F_{q^t}$, where $\alpha \in \fqn^*$.
If $t=1$, we will just write $\mathbb{P}(\fqn)$ and use $\overline{\alpha}:=\alpha\F_q$ for its points.
Suppose $\mathrm{Stab}(U)=\F_{q^t}^*$. 
Then $\alpha U=\beta U\Longleftrightarrow\alpha\beta^{-1}\in \F_{q^t}$, and therefore we have
a one-to-one correspondence between the points of $\mathbb{P}_t(\fqn)$ 
and the elements of $\mathrm{Orb}(U)$. 

For the following definition, we refer to \cite[Definition 2.5, Definition 2.6]{heideweight}. 
Note that the defined notions only depend on $\Orb(U)$ and not on its generating subspace~$U$.

\begin{definition}\label{D-WeightInterDistr}
Let $U\in\cG_q(k,n)$. Define 
\[\omega_{2i}=|\lbrace \alpha U\in\Orb(U)\colon \alpha\in\fqn^*,\, d(U,\alpha U)=2i\rbrace|\] 
for $i\in \{1,\dots,k\}$. We call $(\omega_2,\dots,\omega_{2k})$ \textbf{distance distribution} of $\Orb(U)$.
Furthermore, the minimum distance is $d(\Orb(U))=2(k-\ell)$, where 
\[\ell =\max \{\dim_{\fq}(U\cap \alpha U) \colon \alpha \in\fqn^*\setminus\mathrm{Stab}(U)\}.\]
We call $\ell$ the \textbf{maximum intersection dimension} of $\Orb(U)$. 
For $i \in \{0,\ldots,\ell\}$ we define $\lambda_i=|\mathcal{L}_i|$, where
\[\mathcal{L}_i=\{ \overline{\alpha}\in \mathbb{P}(\fqn) \colon \dim_{\fq}(U\cap \alpha U)=i \},\]
and call $(\lambda_0,\ldots,\lambda_{\ell})$ the \textbf{intersection distribution} of $\Orb(U)$.
Note that if $\text{Stab}(U)=\F_{q^t}^*$, then $\lambda_i=(q^t-1)/(q-1)\omega_{2(k-i)}$ for all $i\in \{0,\ldots,\ell\}$.
\end{definition}

One can easily see that a full-length orbit code generated by a $k$-dimensional subspace has distance at most $2k-2$; see also 
\cite[Remark 2.7]{heideweight}.

\begin{definition}[\mbox{see also \cite[Definition 2.8]{heideweight}}]
Let $U\in\cG_q(k,n)$ and suppose that $\cO=\Orb(U)$  is a full-length orbit code.
Then $\cO$ is called an \textbf{optimal full-length orbit code}  if $d(\cO)=2k-2$, and it is 
a \textbf{quasi-optimal full-length orbit code} if $d(\cO)=2k-4$.
\end{definition}

Recall that a subspace $U\in\cG_q(k,n)$ is a \textbf{Sidon space} if for all $a,b,c,d\in U\setminus\{0\}$ the identity $ab=cd$ implies $\{a\F_q,b\F_q\}=\{c\F_q,d\F_q\}$. 
The interest in Sidon spaces increased especially after \cite{Roth} in which the authors proved the following equivalence between optimal full-length orbit codes and Sidon spaces.

\begin{theorem}\cite[Lemma 34]{Roth}\label{lem:charSidon}
Let $U \in\cG_q(k,n)$. Then~$U$ is a Sidon space if and only if $\mathrm{Orb}(U)$ is an optimal full-length orbit code.
\end{theorem}

In the next section we will see, via a duality argument,  that we can always assume $k\leq n/2$ for the dimension~$k$ of a generator of an orbit code. 

\section{A bound on $k$ for cyclic orbit codes with prescribed minimum distance}\label{S-bound}

In this section, we provide an upper bound on the dimension~$k$ of the subspace $U\in \cG_q(k,n)$ depending on the minimum distance 
of $\Orb(U)$.
Moreover, we recall a result from \cite{CPZ2024IEEE} classifying the cyclic orbit codes generated by $3$-dimensional spaces.

We start with the definition of the orthogonal complement of an $\fq$-subspace of $\fqn$.
Recall the trace function from $\fqn$ to $\fq$, defined as
\[
      \Tr_{q^n/q}: a\in\fqn \longmapsto \sum_{i=0}^{n-1} a^{q^i}\in\fq.
\]
It is an $\fq$-linear map and gives rise to a nondegenerate, symmetric, bilinear form
\[
                 \sigma \colon (a,b)\in\fqn\times\fqn \longmapsto \Tr_{q^n/q}(ab).
\]
For any subset $U\subseteq\fqn $ the orthogonal complement $U^{\perp}$ is defined as
\[
                 U^{\perp}=\{a\in\fqn\colon \Tr_{q^n/q}(ab)=0 \text{ for all }b\in U\}.
\]
As for any nondegenerate, symmetric, bilinear form we have $\dim_{\fq}(U^{\perp})=n-\dim_{\fq}(U)$ and 
$(U+W)^{\perp}=U^{\perp}\cap W^{\perp}$.
Furthermore, the following can easily be verified.

\begin{proposition}
\label{prop:orthcompl}
Let $U,W$ be $\fq$-subspaces of $\fqn$ and $\alpha\in\fqn^*$. Then
   \begin{enumerate}[label=({\arabic*}), leftmargin=1.5em,labelwidth=1.5em,labelsep=0em,align=left, topsep=-1.3ex]
    \item $(\alpha U)^{\perp}=\alpha^{-1}U^{\perp}$;
    \item $d(U^\perp,W^\perp)=d(U,W)$;
    \item $d(U^{\perp},\alpha^{-1}U^{\perp})=d(U,\alpha U)$.
\end{enumerate} 
As a consequence, $(\mathrm{Orb}(U))^{\perp}=\mathrm{Orb}(U^{\perp})$, and the cyclic orbit codes $\mathrm{Orb}(U)$ and $\mathrm{Orb}(U^\perp)$ have the same length and distance distribution.
\end{proposition}

Let $U\in\cG_q(k,n)$. If $d(\mathrm{Orb}(U))=2k-4$, then 
\[
\dim_{\fq}(U\cap \alpha U)\leq 2,
\]
for every $\alpha \in \fqn$ such that $U\neq \alpha U$.
This immediately implies that
\[n\geq \dim_{\fq}(U + \alpha U)=2k-\dim_{\fq}(U\cap \alpha U)\geq 2k-2,\]
that is 
\[k \leq \frac{n+2}2.\]
We can improve the bound as follows.

\begin{proposition} \label{boundquasioptimal}
Let $U\in\cG_q(k,n)$, and suppose that $d(\mathrm{Orb}(U))=2(k-\ell)>2$. 
Then 
\[
   k\leq\frac{n+\ell}{2}.
\]
Furthermore, if $\mathrm{Orb}(U)$ is a full-length orbit code, then $k<\frac{n+\ell}{2}$.
In this case we have in particular
\begin{enumerate}[label=(\alph*),leftmargin=1.5em,labelwidth=1.5em,labelsep=0em,align=left, topsep=-1.3ex]
    \item If $\ell=1$ then $k\leq\frac{n}{2}$.
    \item If $\ell=2$ and $n$ is odd, then $k\leq \frac{n+1}{2}$. Moreover, $k=\frac{n+1}{2}\,$ $\Longleftrightarrow U^\perp$ is a Sidon space.
\end{enumerate}
Note that in Case~(a) $\mathrm{Orb}(U)$ is an optimal full-length orbit code (and~$U$ is a Sidon space), while in Case~(b) it is a  quasi-optimal full-length orbit code.
\end{proposition}

\begin{proof}
$d(\mathrm{Orb}(U))=2(k-\ell)$ implies $\dim_{\fq}(U\cap \alpha U)\leq \ell$ for all $\alpha \in \fqn$ such that $U\neq \alpha U$.
This leads to $n\geq \dim_{\fq}(U + \alpha U)=2k-\dim_{\fq}(U\cap \alpha U)\geq 2k-\ell$,
and thus
\[k \leq \frac{n+\ell}2.\]
Suppose now that $\mathrm{Orb}(U)$ is a full-length orbit code and $k=\frac{n+\ell}{2}$. Then $\dim_{\fq}(U^\perp):=\hat{k}=\frac{n-\ell}{2}$. 
By \cref{prop:orthcompl} $d(\mathrm{Orb}(U^\perp))=d(\mathrm{Orb}(U))=2(k-\ell)=n-\ell=2\hat{k}$. 
Thus $\mathrm{Orb}(U^\perp)$ is a spread. 
In particular, $\hat{k}$ divides~$n$ and $U^\perp$ is a multiplicative coset of $\F_{q^{\hat{k}}}$ in $\F_{q^n}$; see 
\cite[Remark~2.7]{heideweight}.
But this means that $|\mathrm{Orb}(U)|=|\mathrm{Orb}(U^\perp)|=(q^n-1)/(q^{\hat{k}}-1)$.
Thus, if $\mathrm{Orb}(U)$ is a full-length orbit code, then $\hat{k}=1,\,k=n-1$, and  $\ell=n-2$.
Hence $d(\mathrm{Orb}(U))=2(k-\ell)=2$, contradicting our assumption.
All of this shows that $k<\frac{n+\ell}{2}$.
\\
Let us now turn to the special cases (a) and (b) assuming that $\mathrm{Orb}(U)$ is a full-lenght orbit code.
(a) is clear.
(b) We know already that $k\leq\frac{n+1}{2}$. Next, $k=\frac{n+1}{2}\Longleftrightarrow\dim_{\fq}(U^\perp)=\hat{k}=\frac{n-1}{2}$ and 
$2(k-2)=n-3=2\hat{k}-2$.
Hence $d(\mathrm{Orb}(U))=2(k-2)$ is equivalent to $\mathrm{Orb}(U^\perp)$ being an optimal full-length orbit code.
The latter is equivalent to $U^\perp$ being a Sidon space thanks to \cref{lem:charSidon}.
\end{proof}

Let us briefly discuss the existence of optimal and quasi-optimal full-length orbit codes of dimension $n/2$ and $(n-1)/2$.

\begin{remark}\label{R-n2}
\begin{enumerate}[label=(\alph*),leftmargin=1.5em,labelwidth=1.5em,labelsep=0em,align=left, topsep=-1.3ex]
\item Let $n$ be even and $k=n/2$. 
        In \cite[Thm.~12 and Thm.~16]{Roth} Sidon spaces in $\cG_q(n/2,n)$ are constructed for all $q\geq 3$.
       For $q=2$, it is known that there exist optimal full-length orbit codes for $n=4,\,6$ and that there does not exist 
            an optimal full-length orbit code for $n=8$; see \cite[Example~4.9]{GLTroha15}.
   \item Let~$n$ be odd and $k=(n+1)/2$. As we have seen above, the existence of optimal full-length orbit codes is equivalent to the existence of 
           a Sidon space of dimension $k=(n-1)/2$ in $\fqn$. \\
           For $n=5$ and $k=2$, any $2$-dimensional $\fq$-subspace $U$ of $\F_{q^5}$ defines an optimal full length orbit code (see Remark \ref{rmk:k=2nooptimal}) and so by Theorem \ref{lem:charSidon} $U$ is a Sidon space.\\
           For $n=7$, $k=3$ we can consider $U=\langle 1, \lmb, \lmb^3\rangle_{\fq}$ where $\F_{q^7}=\fq(\lmb)$. Indeed, note that $U^2:=\langle uv: u,v\in U\rangle_{\fq}=\langle 1, \lmb, \lmb^2,\lmb^3,\lmb^4,\lmb^6\rangle_{\fq}$, so $\dim_{\fq}(U^2)=6=\min\lbrace n, \binom{k+1}{2}\rbrace$. Thus by \cite[Lemma 20]{Roth}, $U$ is a Sidon space.
           \\
           For $n=9,\,k=4$, and $q=2,3$ random searches easily produce Sidon spaces,
           but as to our knowledge no constructions are known. 
\end{enumerate}
\end{remark}

For dimension less than~$3$, the orbit codes are easily understood.
In particular, there do not exist quasi-optimal full-length orbit codes in $\mathcal{G}_q(1,n)$ and $\mathcal{G}_q(2,n)$. 

\begin{remark}
\label{rmk:k=2nooptimal}
Note that every subspace $U$ in $\cG_q(1,n)$  generates a full-length equidistant code with minimum distance~$2$. 
Let now $U\in\mathcal{G}_q(2,n)$. Then 
\begin{enumerate}[label=(\alph*),leftmargin=1.5em,labelwidth=1.5em,labelsep=0em,align=left, topsep=-1.3ex]
\item either $U$ is a cyclic shift of $\F_{q^2}$ (and $n$ is even) and $\cO:=\mathrm{Orb}(U)=\mathrm{Orb}(\F_{q^2})$, thus $d(\cO)=4=2k$, i.e. $\cO$ is a {\it spread code} (see \cite{manganiello2008spread});
\item or $U$ is not a cyclic shift of $\F_{q^2}$, in which case $d(\mathrm{Orb}(U))=2=2k-2$, i.e. $\mathrm{Orb}(U)$ is an optimal full-length orbit code.  
\end{enumerate}
\end{remark}

In \cite{CPZ2024IEEE}, the cyclic orbit codes of $\mathcal{G}_q(3,n)$ are completely described. In particular, the authors proved the following classification result.

\begin{theorem}\cite[Corollary III.8]{CPZ2024IEEE}\label{T-k3}
\label{thm:classificationcase3dim}
Let $\mathcal{C}=\mathrm{Orb}(U)\subseteq \mathcal{G}_q(3,n)$. We have the following cases.
\begin{enumerate}[label=(\Roman*),leftmargin=2em,labelwidth=2em,labelsep=0em,align=left, topsep=-1.3ex]
    \item $d(\cC)=6$ and $\cC=\mathrm{Orb}(\F_{q^3})$ where $3 \mid  n$.
    \item $d(\cC)=4$ and $\cC$ is an optimal full-length code, $n\geq 6$.
    \item $d(\cC)=2$ and $\cC=\mathrm{Orb}(U)$ is a quasi-optimal full-length code. In particular, one of the following occurs:
    \begin{itemize}[leftmargin=3em,labelwidth=3em,labelsep=0em,align=left, topsep=-1.3ex]
        \item [(III.1)] $U=\la 1,\lmb,\lmb^2\ra_{\fq}$ for some $\lmb \in \F_{q^4}\setminus \F_{q^2}$, $n\geq 4$;
        \item [(III.2)] $U=\la 1,\lmb,\lmb^2\ra_{\fq}$ for some $\lmb \in \F_{q^n}\setminus \F_{q^4}$, $n\geq 5$;
        \item [(III.3)] $U=\F_{q^2}+\la \mu\ra_{\fq}$ for some $\mu \in \F_{q^n}\setminus \F_{q^4}$, $n\geq 5$.
    \end{itemize}
\end{enumerate}   
\end{theorem}

In the next section we will see that the codes in Case~(II) have all the same distance distribution. 
On the other hand, by \cite[Thm.~4.4 and Cor.~4.11]{CPZ2024FullWeight} the 3 classes of codes in Case~(III) have different distance distribution. 
More precisely,  in combination with \cite[Thm.~4.1]{heideweight} (see also \cref{thm:thm4.1heide}) one derives 
\begin{equation}
\label{eq:weightdistribcase3dim}
    (\omega_2,\omega_4,\omega_6)=\left\{\begin{array}{ll}
      \big(q+q^2(q+1),\,0,\,\frac{q^n-q^4}{q-1}\big),&\text{in Case (III.1)},\\
      \big(q(q+1),\,q^3(q+1),\,\frac{q^n-q^5}{q-1}\big),&\text{in Case (III.2)}, \\
      \big(q,\,q^2(q+1)^2,\,\frac{q^n-1}{q-1}-q^2(q+1)^2-q-1\big),&\text{in Case (III.3)}.
      \end{array}\right.
\end{equation}

\section{Fundamental properties of the intersection distribution}

In \cite{heideweight} the intersection distribution of a quasi-optimal cyclic orbit code $\Orb(U)$ has been described. 
It has been shown that the distribution depends  (a) on the number of fractions, $f_U$, that can be formed by the elements in~$U$, 
and (b) on the fact whether or not~$U$ contains a cyclic shift of the subfield $\F_{q^2}$ (if~$n$ is even); see  Equation~\eqref{eq2:lmb2lmb1} and \cref{thm:thm4.1heide}.
Making use of the linear set $L_{U\times U}$, we will derive new bounds for~$f_U$, which in turn will narrow down the possible intersection distributions 
of quasi-optimal full-length orbit codes.

Recall the notation $\overline{\alpha}=\alpha\F_q$ for the points in $\P(\fqn)$.

\begin{definition}
\label{def:parametersfu}{\cite[Definition 3.1]{heideweight}}
Let $U \in\cG_q(k,n)$ and $d(\mathrm{Orb}(U))=2k-2\ell$. We define
\begin{enumerate}[leftmargin=1.5em,labelwidth=1.5em,labelsep=0em,align=left, topsep=-1.3ex]
    \item $\mathcal{L}=\mathcal{L}(U)=\bigcup_{i=1}^\ell \mathcal{L}_i$, where $\cL_i$ is as in \cref{D-WeightInterDistr};
    \item $\mathcal{S}=\mathcal{S}(U)=\{\overline{\alpha}\in \mathbb{P}(\fqn) \colon \alpha \in \mathrm{Stab}(U)\}$;
    \item $\mathcal{M}=\mathcal{M}(U)=\{ (\overline{u},\overline{v}) \colon u,v \in U\setminus\{0\}, \overline{u}\ne\overline{v} \}$;
    \item $\mathcal{F}=\mathcal{F}(U)=\{ \overline{uv^{-1}} \colon u,v \in U\setminus\{0\} \}$. 
             We call $\mathcal{F}$ the \textbf{set of fractions} of $U$ and set $f_U=|\cF|$.
\end{enumerate}
\end{definition}

The above parameters are related as follows. Recall the intersection distribution $(\lambda_0,\ldots,\lambda_\ell)$ from \cref{D-WeightInterDistr}.
Thus, $\lambda_i=|\cL_i|$.

\begin{proposition}{\cite[Corollary 3.3]{heideweight}} \label{prop:fUlambda}
Let $U \in\cG_q(k,n)$ be such that $\mathrm{Stab}(U)=\F_{q^t}^*$ and let $d(\mathrm{Orb}(U))=2k-2\ell$. 
Set  $s=|\mathcal{S}|$, $Q=|\mathcal{M}|$, thus $Q=(q^k-1)(q^k-q)/(q-1)^2$. 
Then
\[ 
    f_U=s +\sum_{i=1}^\ell \lambda_i\ \text{ and }\ Q=\sum_{i=1}^\ell \frac{q^i-1}{q-1}\lambda_i+\frac{q^k-1}{q-1}(s-1). 
\]
If $\mathrm{Orb}(U)$ is a full-length orbit, we have $s=1$ and
\[ 
    f_U=1+\sum_{i=1}^\ell \lambda_i \ \text{ and }\  Q=\sum_{i=1}^\ell \frac{q^i-1}{q-1}\lambda_i. \]
\end{proposition}

The intersection distribution of cyclic orbit codes has also been studied in \cite{bhaintwal2024}. 
In particular, the following theorem has been proven.

\begin{theorem}\cite[Theorems 4 and 8]{bhaintwal2024}\label{T-Bhaintwal}
Let $U\in\cG_q(k,n)$ be a subspace with stabilizer $\fq^*$. 
\begin{enumerate}[label=(\arabic*),leftmargin=1.5em,labelwidth=1.5em,labelsep=0em,align=left, topsep=-1.3ex]
    \item If $n$ is an odd number, then $\lmb_i$ is a multiple of $q(q+1)$ for all $i=0,\dots, k-1$;
    \item If $n$ is even and $U$ contains $\frac{q^{2m}-1}{q^{2}-1}$ distinct shifts of $\F_{q^{2}}$ (where $m\geq 0$), then
    \[
                \lmb_{2m}=q+sq(q+1)
    \]
    for some non negative $s$. Moreover, $\lmb_i$ is multiple of $q(q+1)$ for $i=0,\ldots, k-1, i\neq 2m$.
\end{enumerate}
\end{theorem} 

Together with Proposition \ref{prop:fUlambda} we obtain the following consequence.

\newpage
\begin{corollary}
    \label{cor:fUconsequence}
Let $U\in\cG_q(k,n)$ be a subspace with stabilizer $\fq^*$. 
\begin{enumerate}[label=(\arabic*),leftmargin=1.5em,labelwidth=1.5em,labelsep=0em,align=left, topsep=-1.3ex]
    \item If $n$ is an odd number, then $f_U-1$ is a multiple of $q(q+1)$;
    \item If $n$ is even and $U$ contains $\frac{q^{2m}-1}{q^{2}-1}$ distinct shifts of $\F_{q^{2}}$, then $f_U-1-q$ is a multiple of $q(q+1)$.
\end{enumerate}
\end{corollary}

Moreover, we have the following lower and upper bound on $f_U$.

\begin{proposition}{\cite[Proposition 3.4]{heideweight}}
With the same notation as in \cref{prop:fUlambda} we have
\[\frac{q^k-1}{q-1}\leq f_U \leq Q+1.\]
Furthermore, 
\begin{enumerate}[label=(\arabic*), leftmargin=1.5em,labelwidth=1.5em,labelsep=0em,align=left, topsep=-1.3ex]
    \item  $f_U=(q^k-1)/(q-1)$ if and only if $U=a \F_{q^k}$, for some $a \in \fqn\setminus \F_{q^k}$.
    \item  $f_U=Q+1$ if and only if $U$ is a Sidon space.
\end{enumerate}
\end{proposition}

In the case where~$U$ is a Sidon space, or equivalently $\mathrm{Orb}(U)$ is an optimal full-length code, the intersection distribution is completely known and only depends on the parameters $q,n,k$.

\begin{theorem}{\cite[Theorem 3.7]{heideweight}}
Let $U \in\cG_q(k,n)$ be a Sidon space. Then the optimal full-length orbit $\mathrm{Orb}(U)$ has intersection distribution $(\lambda_0,\lambda_1)$, where
\[
   \lambda_1=Q=\frac{q^k-1}{q-1}\frac{q^k-q}{q-1}=f_U-1\ \text{ and }\ 
   \lambda_0=\frac{q^n-1}{q-1}-\lambda_1-1.
\]
\end{theorem}

One can check straightforwardly that this is compatible with the information given about $\lambda_0$ and $\lambda_1$ in \cref{T-Bhaintwal}.
For quasi-optimal codes the distance distribution does not just depend on $q,n,k$.

\begin{theorem}\cite[Theorem 4.1]{heideweight}
\label{thm:thm4.1heide}
Let $U\in\cG(k,n)$ generate a quasi-optimal full-length orbit. One of the following cases occurs:
\begin{enumerate}[label=(\arabic*),leftmargin=1.5em,labelwidth=1.5em,labelsep=0em,align=left, topsep=-1.3ex]
    \item $U$ contains a cyclic shift of $\F_{q^2}$ (hence $n$ is even). In this case $U$ has intersection distribution $(\lmb_0,\lmb_1,\lmb_2)$ where\[
        \begin{aligned}
            \lmb_2&=q+rq(q+1)\\
            \lmb_1&=Q-(q+1)\lmb_2=\frac{q^k-1}{q-1}\frac{q^k-q}{q-1}-(q+1)(q+rq(q+1)),\\
            \lmb_0&=|\mathbb{P}(\fqn)|-\lmb_1-\lmb_2-1=\frac{q^n-1}{q-1}+q^2(1+r(q+1))-Q-1.
        \end{aligned}
        \]
    for some $r\geq 0$.
    \item $U$ does not contain a cyclic shift of $\F_{q^2}$. In this case $U$ has intersection distribution $(\lmb_0',\lmb_1',\lmb_2')$ where
    \[
        \begin{aligned}
            \lmb_2'&=r'q(q+1)\\
            \lmb_1'&=Q-(q+1)\lmb_2'=\frac{q^k-1}{q-1}\frac{q^k-q}{q-1}-r'q(q+1)^2,\\
            \lmb_0'&=|\mathbb{P}(\fqn)|-\lmb_1'-\lmb_2'-1=\frac{q^n-1}{q-1}+r'q^2(q+1)-Q-1.
        \end{aligned}
    \]
    for some $r'\geq 1.$
\end{enumerate}
\end{theorem}

\begin{remark}
Note that if $\mathrm{Orb}(U)$ and $\mathrm{Orb}(V)$ are quasi-optimal full-length codes where $U$ belongs to Family (1) of Theorem \ref{thm:thm4.1heide} and $V$ belongs to Family (2), then $\mathrm{Orb}(U)$ and $\mathrm{Orb}(V)$ have different distance distributions because $\lmb_2\neq \lmb_2'$ (regardless of~$r$ and~$r'$). In particular, the orbit codes are not isometric in the sense of \cref{D-Isometric}.
\end{remark}

\subsection{New bounds on $f_U$}
In this subsection, we point out a relation between the distance distribution of $\mathrm{Orb}(U)$ and the points of $L_{U\times U}$. This will allow us to prove better bounds on $f_U=|\cF|$; see \cref{def:parametersfu}.
Recall from Section~\ref{S-LinSets} the $(U\!\!\times\!\!U)$-weight distribution $(N_0,\ldots,N_{2k})$.

\begin{proposition}\label{prop:weightandweightpoints}
Let $U \in\cG_q(k,n)$, $\mathcal{O}=\mathrm{Orb}(U)$, and $\mathrm{Stab}(U)=\F_{q^t}^*$. Let $L_{U\times U}$ be the linear set defined by $U\times U$. 
Then $N_i=0$ for $i>k$,   $N_k=q^t+1$ and
\[
    \omega_{2i}=\left\{\begin{array}{cl} 1,&\text{if }i=0,\\  \frac{N_{k-i}}{q^t-1},&\text{if }i=1,\ldots,k\end{array}\right.
    \quad \text{ and }\quad
    \lambda_i=\left\{\begin{array}{cl}\frac{N_{i}}{q-1},&\text{if }i=0,\ldots,k-1,\\ \frac{q^t-1}{q-1},&\text{if }i=k.\end{array}\right.
\]
Moreover, 
\[ 
    f_U=\frac{|L_{U\times U}|-2}{q-1}. 
\]
\end{proposition}

\begin{proof}
We know from \eqref{e-Weight01} and \eqref{e-lessk} that $\la(1,0)\ra_{\fqn}$ and $\la(0,1)\ra_{\fqn}$ have weight~$k$ and
$w_{U\times U}(\la (1,\alpha) \ra_{\fqn})$ 
$=\dim_{\fq}(U\cap\alpha U)$ for all $\alpha\in\fqn^*$.
This shows immediately that $N_k=|\mathrm{Stab}(U)|+2=q^t+1$.
Using that $\omega_{2i}=|\left\{ \alpha U \colon \dim_{\fq}(U\cap \alpha U)=k-i \right\}|$, we obtain 
$N_{k-i}=|\mathrm{Stab}(U)|\omega_{2i}$ for $i=1,\ldots,k$.
This proves the first part. 
The statement about $\lambda_i$ follows immediately from the relation between $\lambda_i$ and $\omega_{2(k-i)}$ as described in 
\cref{D-WeightInterDistr}.
Finally, note that we have a well-defined surjective map $L_{U\times U}\setminus\{\la(1,0)\ra_{\fqn},\la(0,1)\ra_{\fqn}\}\longrightarrow\cF$ given by 
$\subspace{(u_1,u_2)}\longmapsto \overline{u_1u_2^{-1}}$. Its fibers have size $q-1$ and thus 
 we arrive at $f_U=(|L_{U\times U}|-2)/(q-1)$, as stated.
 \end{proof}

\begin{corollary}
\label{cor:w2k-2neq0thenfU}
Let $U \in\cG_q(k,n)$ with $n\geq 2k-1$. 
If $\omega_{2(k-1)}\ne 0$, then
\[ f_U\geq \frac{q^{2k-1}-1}{q-1}. \]
\end{corollary}

\begin{proof}
Since $\omega_{2(k-1)}\ne 0$, \cref{prop:weightandweightpoints} implies $N_1>0$, i.e. $L_{U\times U}$ has at least one point of weight one. As a consequence, by Theorem \ref{th:sizelinset} and Remark \ref{rmk:sizelinsetk=n+1} the size of $L_{U\times U}$ is at least $q^{2k-1}+1$ and so the assertion follows from \cref{prop:weightandweightpoints} .
\end{proof}

\begin{remark}
Note that for $k=3$ the lower bound of Corollary \ref{cor:w2k-2neq0thenfU} is tight for any $n\geq 5$ and $q\geq 2$. Indeed, consider $U\in\cG_q(3,n)$ as in Theorem \ref{thm:classificationcase3dim} Case (III.2). From Equation \eqref{eq:weightdistribcase3dim} 
and \cref{prop:fUlambda} we obtain $f_U=\lmb_1+\lmb_2+1=\frac{q^5-1}{q-1}$.
\end{remark}

We can remove the assumption $\omega_{2(k-1)}\neq 0$ in Corollary \ref{cor:w2k-2neq0thenfU} to derive a new bound, which makes use of the following recent result.

\begin{theorem} {(see \cite[Theorem 2]{csajbok2023maximum})} \label{th:geometricfieldoflinearity}
Let $L_W$ be an $\F_q$-linear set of $W$-rank $h$ in $\PG(r-1,q^n)$ with $n\leq q$ and  $h\leq (r-1)n$.
Let $e=\min\{w_{W}(P) \colon P \in L_W \}$ and suppose $e>1$. Then there exists an integer $t \geq e$ such that $t \mid n$ and $L_W=L_V$ where $V=\langle W \rangle_{\F_{q^t}}$.
\end{theorem}

\begin{corollary} \label{cor:geometricfieldoflinearity}
Let $n\leq q$ and $L_W$ be an $\F_q$-linear set of $W$-rank $h$ in $\PG(1,q^n)=\PG(V,\fqn)$ where  $h\leq n$ and $|L_W|>1$.
Let $e=\min\{w_{W}(P) \colon P \in L_W \}$. Then there exists an integer $t\in\{e,\ldots,n-1\}$ such that $t\mid n$ and
\[
    |L_W| \geq q^{h-t}+1.
\]
Furthermore, if $e=1$ or $n$ is a prime number, then $t=1$.
\end{corollary}

\begin{proof}
If $e=1$, Theorem \ref{th:sizelinset} implies the assertion for $t=1$. 
Thus, let $e>1$.
Then by Theorem \ref{th:geometricfieldoflinearity} there exists an integer $t_1\in\{e,\ldots,n\}$ such that  $t_1\mid n$ and $L_W=L_{V_1}$ where $V_1=\langle W \rangle_{\F_{q^{t_1}}}$.
\\
i) Suppose $t_1=n$. 
Then $V_1=\subspace{W}_{\fqn}$ and $\dim_{\fqn}(V_1)\in\{1,2\}$.
If $\dim_{\fqn}(V_1)=1$, then $L_W=L_{V_1}$ consists just of the point $V_1$, contradicting the assumption that $|L_W|>1$.
If $\dim_{\fqn}(V_1)=2$, then $L_W=L_{V_1}=\mathrm{PG}(1,q^n)$, a contradiction to $h\leq n$. 
\\
ii) Thus $t_1<n$. 
Note that $L_{V_1}$ is also an $\F_{q^{t_1}}$-linear set and 
\[
    h_1:=\dim_{\F_{q^{t_1}}}(V_1)\geq \frac{\dim_{\fq}(W)}{t_1}=\frac{h}{t_1}.
\]
For any $w\in W$ and $P=\langle w\rangle_{\fqn}\in L_W$ we have
\[			
    w_{W}(P)=\dim_{\fq}(\subspace{w}_{\fqn}\!\cap\! W)\leq \dim_{\fq}(\subspace{w}_{\fqn}\!\cap\! \langle W\rangle_{\F_{q^{t_1}}})
        = t_1\dim_{\F_{q^{t_1}}}(\subspace{w}_{\fqn}\!\cap\! V_1)=t_1w_{V_1}(P).
\]
Therefore $e\leq w_{W}(P)\leq t_1 w_{V_1}(P)$ for any $P=\langle w\rangle_{\fqn}\in L_W$.
Let $e_2=\min\{w_{V_1}(P) \colon P \in L_{V_1} \},$ where $L_{V_1}$ is considered as an $\F_{q^{t_1}}$-linear set in $\PG(1,(q^{t_1})^{\frac{n}{t_1}})$. If $e_2=1$, then by Theorem  \ref{th:sizelinset} we get
\[
    |L_W|=|L_{V_1}|\geq (q^{t_1})^{h_1-1}+1\geq q^{h-t_1}+1,
\]
and the assertion is proved with $t=t_1$. 
Let $e_2>1$.
Since $h_1\leq \frac{n}{t_1}$ and $\frac{n}{t_1}\leq q^{t_1}$, we may apply \cref{th:geometricfieldoflinearity} again and obtain an integer 
$t_2\in\{e_2,\ldots,\frac{n}{t_1}-1\}$ such that  $t_2\mid\frac{n}{t_1}$ and  $L_W=L_{V_1}=L_{V_2}$, where 
$V_2=\subspace{V_1}_{\F_{q^{t_1t_2}}}$.
Furthermore,
\[
    h_2=\dim_{\F_{q^{t_1t_2}}}(V_2)\geq \frac{\dim_{\F_{q^{t_1}}}(V_1)}{t_2}=\frac{h_1}{t_2} \geq \frac{h}{t_1t_2}.
\]
Proceeding in this fashion, we arrive after a finite number $s$ of steps at a sequence $t_1$,..., $t_s$ such that 
$t_1t_2\cdots t_s \mid  n$ and $L_W=L_{V_1}=....=L_{V_s}$,  where $V_i=\subspace{V_{i-1}}_{\F_{q^{t_1t_2\ldots t_i}}}$
and $e_s=\min\{w_{V_s}(P) \colon P \in L_{V_s} \}=1,$ and where $L_{V_s}$ is considered as an $\F_{q^{t_1t_2\cdots t_s}}$-linear set.
Moreover,
\[
   h_s=\dim_{\F_{q^{t_1t_2\cdots t_s}}}(V_s)\geq \frac{h_{s-1}}{t_s} \geq \frac{h}{t_1t_2\cdots t_s}.
\]
Now applying Theorem \ref{th:sizelinset} to $L_{V_s}$ we obtain
\[
    |L_W|=|L_{V_s}|\geq (q^{t_1t_2\cdots t_s})^{h_s-1}+1\geq q^{h-t}+1,
\]
where $t=t_1t_2\cdots t_s$.
This concludes the proof.
\end{proof}

Now we are ready for the following theorem, which provides a better upper bound on the parame\-ter~$f_U$ than \cite[Proposition 3.4]{heideweight}. 
This also answers an open problem posed in \cite[Section 6]{heideweight}.

\begin{theorem} \label{th:bounds f_U}
Let $k\leq n/2$ and $U \in\cG_q(k,n)$. 
Let $(\omega_2,\ldots,\omega_{2k})$ be the distance distribution of $\Orb(U)$ and set 
$l=\max\{i \in \{1,\ldots,k-1\}  \colon  \omega_{2i}\ne 0\}$. Then
\[
    f_U\leq \frac{(q^{2k}-1)-(|\text{Stab}(U)|+2)(q^k-q^{k-l})-2(q^{k-l}-1)}{(q^{k-l}-1)(q-1)}.
\]
Moreover, if $q\geq n$, then there exists a divisor~$s$ of~$n$ such that $k-l\leq s< n$ and
\begin{equation}\label{eq:boundf2} 
    f_U\geq \frac{q^{2k-s}-1}{q-1}.
\end{equation}
In particular,  if  $l=k-1$ or $n$ is prime, then $s=1$. 
\end{theorem}

\begin{proof}
Let $W=U\times U$ and consider the $\fq$-linear set $L_W$ of $W$-rank $2k$  in $\PG(1,q^n)$. 
The definition of $\omega_{2j}$ (see \cref{D-WeightInterDistr}) tells us that $l$ is the maximum number less than~$k$ such that there exists
$\alpha\in\fqn^*$ with $\dim_{\fq}(U\cap \alpha U)=2(k-l)$.
Thus \eqref{e-lessk} implies
\[
    k-l= \min\{w_{W}(P) \colon P \in L_W \}.
\]
Let $\mathrm{Stab}(U)=\F_{q^t}^*$.
Then $t\mid k$ and $t\mid(k-l)$.
Consider now $U$ as a $\F_{q^t}$-vector space. As such it has dimension $k/t$.
Moreover, the $\F_{q^t}$-linear set~$L_W$ has minimum $W$-weight $(k-l)/t$.
Applying Proposition \ref{upper boundlinearsetcompwe}  we obtain
\[
    |L_{W}|\leq \frac{(q^{2k}-1)-(q^t+1)(q^k-q^{k-l})}{(q^{k-l}-1)}
\]
and the first part of the assertion follows by Proposition \ref{prop:weightandweightpoints}. 
The second part follows by  Corollary \ref{cor:geometricfieldoflinearity} applied to the $\F_{q}$-linear set $L_W$ together with Proposition \ref{prop:weightandweightpoints}. 
\end{proof}

We now return to the intersection distribution $(\lambda_0,\lambda_1,\lambda_2)$ of quasi-optimal orbit codes as described in \cref{thm:thm4.1heide}.
We derive necessary conditions for the case $\lmb_1=0$. 
Furthermore, in the case where $\lmb_1\neq 0$ we will make use of \cref{cor:w2k-2neq0thenfU} in order to 
provide an upper bound on the parameter $\lmb_2$, which improves upon the one given in \cite{heideweight}.

\begin{corollary}
\label{cor:fUconsequencequasi-optimal}
Let $U\in\cG_q(k,n)$ generate a quasi-optimal full-length orbit code and let $(\lmb_0,\lmb_1,\lmb_2)$ be its intersection distribution. Moreover, let $\epsilon=1$ if $U$ contains a cyclic shift of $\F_{q^2}$ and $\epsilon=0$ otherwise.
We have the following cases.
\begin{enumerate}[label=(\arabic*),leftmargin=1.5em,labelwidth=1.5em,labelsep=0em,align=left, topsep=-1.3ex]
    \item[(a)] $\lmb_1=0$. In this case
        \begin{itemize}
            \item[(a.1)]\  if $k=3$, then $\epsilon=1$ and $\lmb_2=q+q^2(q+1)$.
            \vspace{0.1cm}
            \item[(a.2)]\  if $k>3$, then $(q+1)\mid\big(\lfloor\frac{k}{2}\rfloor-\epsilon\big)$.   
        \end{itemize}
            \vspace{0.2cm}
    \item[(b)] $\lmb_1\neq 0$. In this case $\lmb_2=\epsilon q+rq(q+1)$ for some
         integer $r$ such that
    \[
                  1-\epsilon\leq r\leq \frac{(q^{k-1}-1)(q^{k-2}-1)}{(q+1)(q-1)^2}-\epsilon.
    \]
    In particular, if $k=3$, then $\lmb_2=q$ if $U$ contains a cyclic shift of $\F_{q^2}$ and $\lambda_2=q(q+1)$ otherwise.
\end{enumerate}
\end{corollary}

\begin{proof}
By \cref{thm:thm4.1heide} we have that
\begin{equation}
\label{eq2:lmb2}
     \lmb_2=\epsilon q+ rq(q+1)\ \text{ for some }\ r\geq1-\epsilon.
\end{equation}
Furthermore, from \cite[Equation (4.3) and Corollary 4.2]{heideweight} we know that 
\begin{equation}
\label{eq2:lmb2lmb1}
       q\lmb_2=Q-(f_U-1)\ \text{ and }\ q\lmb_1=(q+1)(f_U-1)-Q,\ \text{ where }\ Q=\frac{(q^k-1)(q^k-q)}{(q-1)^2}.
\end{equation}
(a) Let $\lmb_1=0$. Then $f_U-1=\frac{Q}{q+1}$ and \eqref{eq2:lmb2} and  \eqref{eq2:lmb2lmb1} imply
\begin{equation}\label{e-lmb2}
    \lambda_2=\epsilon q+rq(q+1)=\frac{1}{q}(Q-(f_U-1))=\frac{Q}{q+1},
\end{equation}
which in turn leads to
\begin{equation}\label{e-lmb22}
    \frac{(q^{k}-1)(q^{k-1}-1)}{(q^2-1)(q-1)} \equiv \epsilon \, (\mmod{q+1}). 
\end{equation}
If $k=3$,  \eqref{e-lmb2} implies $\lmb_2=\frac{Q}{q+1}=q+q^2(q+1)$ and therefore $\epsilon=1$.
Thus suppose $k>3$.
Set
\[
   t=\big\lfloor\frac{k}{2}\big\rfloor\ \text{ and }\delta=
   \left\{\begin{array}{cl}1,&\text{if $k$ is odd},\\ -1,&\text{if $k$ is even}.\end{array}\right.
\]
In the next computation we make use of the identity $\sum_{i=0}^{t-1}q^{2i}=\sum_{i=0}^{t-2}(t-i-1)q^{2i}(q^2-1)+t$, which can be easily verified.
Using twice that $2t+\delta-1$ is even, we derive
\begin{align*}
    \frac{(q^{k}-1)(q^{k-1}-1)}{(q^2-1)(q-1)}&= \frac{q^{2t}-1}{q^2-1} \frac{q^{2t+\delta}-1}{q-1}   
      =\Big(\sum_{i=0}^{t-1}q^{2i}\Big)\Big(\sum_{i=0}^{2t+\delta-1}q^i\Big)\\
      &=\Big(\sum_{i=0}^{t-2}(t-i-1)q^{2i}(q^2-1)+t\Big)\Big(q^{2t+\delta-1}+(q+1)\sum_{i=0}^{2t+\delta-3}q^i\Big)\\
      &\equiv t q^{2t+\delta-1}\, (\mmod{q+1})\equiv t\,  (\mmod{q+1}).
\end{align*}
Now \eqref{e-lmb22} implies $\epsilon\equiv t \  (\mmod{q+1})$, which proves the statement.
\\
(b) Let $\lmb_1=\omega_{2(k-1)}\neq 0$. Then \cref{cor:w2k-2neq0thenfU} implies $f_U\geq (q^{2k-1}-1)/(q-1)$. 
Hence \eqref{eq2:lmb2} and  \eqref{eq2:lmb2lmb1} lead to
\begin{align*}
   \epsilon q^2+rq^2(q+1)&=Q-f_U+1\\
     &\leq Q-\frac{q^{2k-1}-1}{q-1}+1\\
     &=\frac{(q^k-1)(q^k-q)-(q^{2k-1}-1)(q-1)+(q-1)^2}{(q-1)^2} =\frac{q^2(q^{k-1}-1)(q^{k-2}-1)}{(q-1)^2},
\end{align*}
from which we get
\[
r\leq \left\lfloor\frac{(q^{k-1}-1)(q^{k-2}-1)}{(q-1)^2(q+1)}-\frac{\epsilon}{(q+1)}\right\rfloor.
\]
Observing that $\frac{(q^{k-1}-1)(q^{k-2}-1)}{(q^2-1)(q-1)}$ is an integer, we obtain the desired upper bound.
\end{proof}

\begin{remark}\label{R-k3Case}
The previous corollary implies that if $\mathrm{Orb}(U)$ is a quasi-optimal full-length orbit code with $k=3$, then $\lambda_2=q,\,q(q+1)$, or $q+q^2(q+1)$, and the latter case implies $\lambda_1=0$. 
\cite[Table~2]{heideweight} shows that all cases do indeed occur.
Furthermore, for $k=4$ and $\epsilon=1$ we have $r\leq q^2+q$. 
Thus, if $q=2$, the maximum value for $\lambda_2$ is $38$ (if $\lambda_1\neq0$), and this does indeed occur; see again \cite[Table~2]{heideweight}
(note  that in this case $\lmb_1=96\neq 0$). 
\end{remark}

\begin{remark}
Let us briefly compare the upper bound on $\lmb_2$ obtained in \cref{cor:fUconsequencequasi-optimal} with that in \cite[Corollary 4.2]{heideweight}.
The latter tell us that
\[
      \lmb_2\leq \frac{Q}{q+1}=\frac{(q^k-q)(q^k-1)}{(q-1)^2(q+1)}
\]
(with equality if and only if $\lambda_1=0$).
On the other hand, if $\lambda_1\neq0$, \cref{cor:fUconsequencequasi-optimal} implies 
\[
\lmb_2\leq \frac{(q^k-q)(q^{k-2}-1)}{(q-1)^2}-\epsilon q^2,
\]
One easily verifies that $q^{k-2}-1\leq (q^k-1)/(q+1)$ for all $q\geq2$, and therefore the second bound is better than the former one.
\end{remark}

\section{Constructions of quasi-optimal codes}
In this section we construct families of quasi-optimal full-length orbit codes in $\cG_q(k,2k)$. 
As a consequence we prove that when $n$ is even, quasi-optimal full-length orbit codes in $\cG_q(k,n)$ exist for any $q\geq 2$ 
and any $3\leq k\leq \frac{n}{2}$.
More precisely, we will study orbits generated by subspaces of the form $U_{s,\gamma}=\lbrace u+u^{q^s}\gamma\colon u\in\fqk\rbrace\in\cG_q(k,n)$, where~$k$ is a divisor of~$n$,\, $s\in\{1,\ldots,k-1\}$ such that $\gcd(s,k)=1$, and $\gamma\in\fqn$ such that  
$\fqn=\fqk(\gamma)$.

We will need, for any divisor $k$ of~$n$, the field norm 
$\N_{q^n/q^k}: \F_{q^n}^*\longrightarrow\F_{q^k}^*,\ a\longmapsto a^{(q^n-1)/(q^k-1)}$.

\begin{remark}
\label{rmk:UsgammaoptimaliffNormneq1}
In \cite{Roth}, the authors study the subspaces $U_{s,\gamma}$.  
In particular, they investigate when these subspaces define optimal full-length orbit codes, i.e., when $U_{s,\gamma}$ is a Sidon space in $\fqn$;
see Theorem \ref{lem:charSidon}. 
In \cite[Theorem 12]{Roth} it is shown that if $k$ is any divisor of~$n$ less than $n/2$ and $s=1$, then 
$\cO_{s,\gamma}=\Orb(U_{s,\gamma})$ is an optimal full-length orbit code in 
$\cG_q(k,n)$. 
Furthermore, in \cite[Theorem 16]{Roth} it is proven that if $q>2,\,n=2k,\,k>2$ and $\gamma\in\F_{q^{2k}}\setminus\F_{q^k}$ is a root of $x^2+bx+c=0$ for some $b,c\in\fqk$ such that $\N_{q^k/q}(c)\neq 1$, then $\cO_{s,\gamma}$ is an optimal full-length orbit code in $\cG_q(k,2k)$. 
Subsequently, in \cite[Theorem 4.5]{CPSZSidon} the authors prove that for any prime power $q$, any $k>2$ and $n=2k$ the code $\cO_{s,\gamma}$ is an optimal full-length orbit code if and only if $\N_{q^{2k}/q}(\gamma)\neq 1$.
\end{remark}

The previous remark tells us that if $\cO_{s,\gamma}$ is not optimal then $n=2k,\,\N_{q^n/q}(\gamma)=1$, and  $k>2$ (see Remark \ref{rmk:k=2nooptimal}).
We will show in \cref{P-OptQuasiOpt} that in these cases $\cO_{s,\gamma}$ is in fact quasi-optimal. 
In addition, in this section we will investigate when the subspace $U_{s,\gamma}$ contains a cyclic shift of $\F_{q^2}$.
As we saw in \cref{thm:thm4.1heide}, this impacts the distance distribution of $\cO_{s,\gamma}$.

Throughout the rest of this section we fix the following notation.

\begin{definition}\label{D-Usgamma}
Let~$k>2,\,n=2k$, and $s\in\{1,\ldots,k-1\}$ be such that $\gcd(s,k)=1$. Furthermore, let~$\gamma\in\fqn$ be such that  
$\fqn=\fqk(\gamma)$.
Define
\[
U_{s,\gamma}=\lbrace u+u^{q^s}\gamma\colon u\in\fqk\rbrace\in\cG_q(k,n)
\]
and set $\cO_{s,\gamma}=\mathrm{Orb}(U_{s,\gamma})$.
\end{definition}

In order to study the properties of the codes $\cO_{s,\gamma}$ we need some preparation.
The last two parts of the following lemma will be used frequently in the remainder of this section. 
The first part will be needed later in Section~\ref{S-AutGroups}.

\begin{lemma}\label{L-theta}
Let $p$ be a prime and $q=p^h$. Let $t\leq hn$ and define $\theta_{p^t}:\Fqn\rightarrow\Fqn,\ a\mapsto a^{p^t-1}$.
\begin{enumerate}[label=(\alph*),leftmargin=1.5em,labelwidth=1.5em,labelsep=0em,align=left, topsep=-1.3ex]
\item Let~$t$ be a divisor of $hn$. Then $\theta_{p^t}^{-1}(\theta_{p^t}(b))=b\F_{p^t}$ for all $b\in\Fqn$.
\item Let~$k$ be a divisor of~$n$. Then $\theta_q(\F_{q^k})=\theta_{q^s}(\F_{q^k})$ for all $s\in[k]$ such that $\gcd(s,k)=1$.
\item Let~$k$ be a divisor of~$n$. Then $\ker\N_{q^n/q^k}=\theta_{q^k}(\F_{q^n}^*)$. Moreover, $\N_{q^n/q^k}(-1)=(-1)^{n/k}$.
\end{enumerate}
\end{lemma}

\begin{proof}
(a) ``$\supseteq$'' is clear. For the converse let $a\in\theta_{p^t}^{-1}(\theta_{p^t}(b))$. Thus $a^{p^t-1}=b^{p^t-1}$, and this implies
$(a/b)^{p^t}=a/b$. Hence $a/b\in\F_{p^t}$ and $a\in b\F_{p^t}$.
\\
(b) 
Let $\F_{q^k}^*=\subspace{\alpha}$. Then $\theta_{q}(\F_{q^k}^*)=\subspace{\alpha^{q-1}}$ and
$\theta_{q^s}(\F_{q^k}^*)=\subspace{\alpha^{q^s-1}}$. 
Since $\gcd(s,k)=1$, we have $\text{ord}(\alpha^{q-1})=(q^k-1)/(q-1)=\text{ord}(\alpha^{q^s-1})$
and the result follows from the obvious containment $\theta_{q^s}(\F_{q^k}^*)\subseteq\theta_{q}(\F_{q^k}^*)$. 
\\
(c) The statement about the kernel is well known. The identity $\N_{q^n/q^k}(-1)=(-1)^{n/k}$ only needs verification for~$q$ odd.
In that case let $n=\ell k$. 
Then $(q^n-1)/(q^k-1)=1+q^k+\ldots+q^{(\ell-1)k}$, and this is even if and only if $\ell$ is even.
\end{proof}

We also need $\cL_{q^s}(\F_{q^k})=\{f_0X+f_1X^{q^s}+\ldots+f_{t}X^{q^{st}}\colon t\in\NN_0,\,f_i\in\F_{q^k}\}$, that is, $\cL_{q^s}(\F_{q^k})$ is the space of $q^s$-polynomials with coefficients in $\F_{q^k}$.
For $f\in\cL_{q^s}(\F_{q^k})$ we denote by $\hat{f}$ the associated $\F_q$-linear map $\F_{q^k}\longrightarrow\F_{q^k}, x\longmapsto \hat{f}(x):=f(x)$.

\begin{lemma}\cite[Theorem 5]{rootspolynomials},\cite[Theorem 10]{rootspolynomialsnorm} (see also \cite[Lemma 3]{sheekey2016new})
\label{lem:qspolynomialsrootsandnorm}
Let $\gcd(s,k)=1$ and $f=f_0X+f_1X^{q^s}+\ldots+f_{t}X^{q^{st}}\in\cL_{q^s}(\fqk)$ be of $q^s$-degree $t\geq0$. Then
\[
    \dim_{\F_{q}}\ker\hat{f}\leq t.
\]
Moreover, if $\dim_{\F_{q}}\ker\hat{f}=t$ then $\N_{q^k/q}(f_0)=(-1)^{stk}\N_{q^k/q}(f_t)$.
\end{lemma}

\begin{lemma}\label{L-UalphaU}
Let $U:=U_{s,\gamma}$ be as in \cref{D-Usgamma}.
Fix $\alpha\in\F_{q^{2k}}^*$ and write $\alpha=\alpha_0+\alpha_1\gamma$ with $\alpha_0,\alpha_1\in\fqk$.
Define
\begin{equation}\label{e-falpha}
   f_{\alpha}=-\alpha_1X+(\alpha_0^{q^s}-\alpha_0-\alpha_1A)X^{q^s}+\alpha_1^{q^s}B^{q^s}X^{q^{2s}}\in\cL_{q^s}(\F_{q^k})
\end{equation}
where $A=\Tr_{q^{2k}/q^k}(\gamma)$ and $B=-\N_{q^{2k}/q^k}(\gamma)$.
Then for any $x\in\F_{q^k}$
\[
    \alpha (x+x^{q^s}\gamma)\in U\cap\alpha U \Longleftrightarrow x\in\ker\hat{f}_{\alpha}.
\]
As a consequence, $\dim_{\fq}(U\cap\alpha U)=\dim_{\fq}(\ker\hat{f}_{\alpha})$
and $\hat{f}_\alpha\equiv0$ if and only if $\alpha\in\fq$.
\end{lemma}

\begin{proof}
By assumption on $\gamma$ in \cref{D-Usgamma} the set $\{1, \gamma\}$ is an $\fqk$-basis of $\F_{q^{2k}}$. 
Thus there exists $A,B\in\fqk$ such that $\gamma^2=A\gamma+B$ and in fact $A=\Tr_{q^{2k}/q^k}(\gamma)$ and $B=-\N_{q^{2k}/q^k}(\gamma)$.
Fix   $x\in\F_{q^k}$. Then 
\begin{align*}
   \alpha(x+\gamma x^{q^s})\in U\cap\alpha U
    &\Longleftrightarrow\alpha(x+\gamma x^{q^s})=y+\gamma y^{q^s}\text{ for some }y\in\fqk\\
    &\Longleftrightarrow (\alpha_0+\alpha_1\gamma)(x+\gamma x^{q^s})=y+\gamma y^{q^s}\text{ for some }y\in\fqk\\
    &\Longleftrightarrow \left\{\begin{array}{ccl}
        y&\!\!=\!\!&\alpha_0x+\alpha_1Bx^{q^s},\\
        y^{q^s}&\!\!=\!\!&\alpha_0 x^{q^s}+\alpha_1 x+\alpha_1 A x^{q^s}\end{array}\right\}\text{ for some }y\in\fqk\\
    &\Longleftrightarrow -\alpha_1 x+(\alpha_0^{q^s}-\alpha_0-\alpha_1A) x^{q^s}+\alpha_1^{q^s}B^{q^s}x^{q^{2s}}=0\\
    &\Longleftrightarrow x\in\ker\hat{f}_\alpha.
\end{align*}
As a consequence, the map $\phi: \ker\hat{f}_{\alpha}\longrightarrow U\cap \alpha U,\ x\mapsto\alpha(x+x^{q^s}\gamma)$, 
is an $\fq$-isomorphism and therefore 
$\dim_{\fq}(U\cap\alpha U)=\dim_{\fq}\ker\hat{f}_{\alpha}$.
Finally, using \cref{lem:qspolynomialsrootsandnorm} and the fact that $k>2$, we arrive at
\[
   \hat{f}_{\alpha}\equiv0\Longleftrightarrow\dim_{\F_q}\ker\hat{f}_\alpha=k\Longleftrightarrow f_\alpha=0
   \Longleftrightarrow[\alpha_1=0\text{ and }\alpha_0\in\F_{q^s}\cap\F_{q^k}=\fq]\Longleftrightarrow\alpha\in\fq.
   \qedhere
\] 

\end{proof}

Now we are ready to determine the distance distribution of the codes $\cO_{s,\gamma}\subseteq \cG_q(k,2k), k>2$.
The following result generalizes  \cite[Theorem~16]{Roth} and \cite[Theorem 4.5]{CPSZSidon}; see also
\cref{rmk:UsgammaoptimaliffNormneq1}. 

\begin{proposition}\label{P-OptQuasiOpt}
Let $U_{s,\gamma}$ and $\cO_{s,\gamma}$ be as in \cref{D-Usgamma}. 
Then
\begin{enumerate}[label=(\arabic*),leftmargin=1.5em,labelwidth=1.5em,labelsep=0em,align=left, topsep=-1.3ex]
\item $\cO_{s,\gamma}$ is a full-length orbit code.
\item $d(\cO_{s,\gamma})\in\{2k-4,2k-2\}$, i.e.  $\cO_{s,\gamma}\subseteq \cG_q(k,2k)$ is  quasi-optimal or optimal. 
\item $\cO_{s,\gamma}$ is quasi-optimal $\Longleftrightarrow\N_{q^{2k}/q}(\gamma)=1$.
\item $\cO_{s,\gamma}$ is optimal $\Longleftrightarrow\N_{q^{2k}/q}(\gamma)\neq1$.
\end{enumerate}
As a consequence, for $q=2$, the orbit code $\cO_{s,\gamma}$ is quasi-optimal.
\end{proposition}

\begin{proof}
By \cref{L-UalphaU} and \cref{lem:qspolynomialsrootsandnorm} we have  for all $\alpha\in\F_{q^n}\setminus \fq$ that
$f_\alpha$ as in \eqref{e-falpha} is nonzero and 
\begin{equation}
\label{eq:diminters=dimker}
    \dim_{\fq}(U_{s,\gamma}\cap\alpha U_{s,\gamma})=\dim_{\fq}(\ker\hat{f}_{\alpha})\leq q^s\text{-degree of } f_{\alpha}.
\end{equation}
(1) Suppose there exists $\alpha\in\mathrm{Stab}(U_{s,\gamma})\setminus\fq$. 
Thus $U_{s,\gamma}=\alpha U_{s,\gamma}$.
Hence \eqref{eq:diminters=dimker} implies that the $q^s$-degree of $f_{\alpha}$ is at least~$k$.
Since $k>2$, this contradicts the definition of $f_{\alpha}$.
Thus $\mathrm{Stab}(U_{s,\gamma})=\fq^*$ and $\cO_{s,\gamma}$ is a full-length orbit code. 
\\
(2) Let $\alpha\in\fqn\setminus\fq$. Again, \eqref{eq:diminters=dimker} implies $\dim_{\fq}(U_{s,\gamma}\cap\alpha U_{s,\gamma})\leq 2$. 
Therefore, $d(\cO_{s,\gamma})$ is $2k-2$ or $2k-4$, 
which means that $\cO_{s,\gamma}$ is either quasi-optimal or optimal.
\\
(3) 
Suppose that $\cO_{s,\gamma}$ is quasi-optimal and let $\alpha\in\fqn\setminus\fq$ be such that 
$\dim_{\fq}(U_{s,\gamma}\cap \alpha U_{s,\gamma})=2$. 
Then we have equality in \eqref{eq:diminters=dimker} and 
\cref{lem:qspolynomialsrootsandnorm} implies that $\N_{q^{k}/q}(-\alpha_1)=(-1)^{2sk}\N_{q^{k}/q}(\alpha_1^{q^s}B^{q^s})$.
Noting that $\N_{q^{k}/q}(-\alpha_1)=(-1)^k\N_{q^{k}/q}(\alpha_1)$ we obtain
\[
  (-1)^k=\N_{q^{k}/q}(\alpha_1^{q^s-1})\N_{q^k/q}(B^{q^s})=\N_{q^k/q}(B)=\N_{q^{k}/q}(-\N_{q^{2k}/q^k}(\gamma))=(-1)^k\N_{q^{2k}/q}(\gamma).
\]
Hence $\N_{q^{2k}/q}(\gamma)=1$. 
Conversely, if $\N_{q^{2k}/q^k}(\gamma)=1$, then by \cite[Theorem 4.5]{CPSZSidon}, $U_{s,\gamma}$ is not a Sidon space and 
Theorem \ref{lem:charSidon} implies that $\cO_{s,\gamma}$ is not optimal.
Hence by~(2) $\cO_{s,\gamma}$ is quasi-optimal. 
\\ 
(4) follows from~(3) and~(2).
\end{proof}

Later in \cref{P-CountOpt} we will determine the number of quasi-optimal codes and of optimal codes of the form $\cO_{s,\gamma}$ for any $n=2k$ and~$q$.

As a consequence of \cref{P-OptQuasiOpt}, we prove the existence of quasi-optimal full-length orbit codes in $\cG_q(k,n)$ for any $n$ even and $3\leq k\leq \frac{n}{2}$.

\begin{theorem} \label{thm:existencequasi-optimal}
If $n$ is even, then there exist quasi-optimal full-length orbit codes in 
$\cG_q(k,n)$ for every $3\leq k \leq \frac{n}{2}$ and for every $q\geq 2$.
\end{theorem}

\begin{proof}
Suppose $k=3$ and $n\geq 6$, where $n$ is even. Consider
\[
U=\langle 1, \lmb, \lmb^2\rangle_{\fq}\subseteq \fqn
\]
where $\lmb\in\fqn\setminus (\F_{q^2}\cup \F_{q^3})$. Thus $\dim_{\fq}(U)=3$.
Since $\lambda \notin \F_{q^3}$, we have $\dim_{\fq}(U\cap \alpha U) \in \{0,1,2\}$ for all $\alpha \in\fqn\setminus \F_{q}$.
Hence $\text{Stab}(U)=\fq^*$.
Moreover,  $U\cap\lmb U= \langle \lambda, \lmb^2\rangle_{\fq}$, 
and so $\text{Orb}(U)$ is a quasi-optimal full-length orbit code in $\cG_q(3,n)$.\\
Suppose $4\leq k\leq \frac{n}{2}$. Choose $\gamma\in\fqn\setminus\F_{q^{n/2}}$ such that $\N_{q^n/q}(\gamma)=1$ (which is indeed possible; see also \cref{P-CountOpt}).
Consider the subspace $U:=U_{1,\gamma}\in\cG_q(\frac{n}{2},n)$  as in \cref{D-Usgamma}.
Then its orbit $\cO_{1,\gamma}$ is a quasi-optimal full-length orbit code thanks to \cref{P-OptQuasiOpt}.
This means that there exists $\bar{\alpha}\in\fqn \setminus \fq$ such that $\dim_{\fq}(U\cap\bar{\alpha} U)=2$. 
Define $S=U\cap\bar{\alpha} U$. Then $S= \bar{\alpha} S'$ for some $S'\subseteq U$ and $\dim_{\fq}(S')=\dim_{\fq}(S)=2$.
There exists a subspace $\hat{U}\subseteq U$ such that $k=\dim_{\fq}\hat{U}$ and
$S+S'\subseteq \hat{U}$.
Thus $S= \hat{U}\cap\bar{\alpha}\hat{U}$
and therefore $\dim_{\fq}(\hat{U}\cap \bar{\alpha}\hat{U})=2$.
Moreover, $\dim_{\fq}(\hat{U}\cap\alpha\hat{U})\leq \dim_{\fq}(U\cap\alpha U)\leq 2$ for any $ \alpha \in\fqn\setminus \F_{q}$, and therefore 
$\mathrm{Orb}(\hat{U})$ is a quasi-optimal full-length orbit code in $\cG_q(k,n)$.
\end{proof}

In light of Theorem \ref{thm:thm4.1heide} we want to characterize which spaces $U_{s,\gamma}$ contain a cyclic shift of $\F_{q^2}$. 

\begin{theorem}\label{T-ContainCycShift}
Let $U_{s,\gamma}$ be as in \cref{D-Usgamma}. 
The following are equivalent.
\begin{enumerate}[label=(\roman*),leftmargin=1.8em,labelwidth=1.8em,labelsep=0em,align=left, topsep=-1.3ex]
\item $U_{s,\gamma}$ contains a cyclic shift of $\F_{q^2}$.
\item There exist two elements of $U_{s,\gamma}$ such that their ratio is in $\F_{q^2}\setminus\fq$.
\item There exist $\eta\in\F_{q^2}\setminus\fq$ and $u,\,v\in\F_{q^k}^*$ such that $u/v\not\in\fq$ and 
      \[
           \gamma=\frac{u-\eta v}{-u^{q^s}+\eta v^{q^s}}.
      \]
\item $k$ is odd and $\N_{q^n/q^2}(\gamma)=-1$.
\end{enumerate}
\end{theorem}

\begin{proof} 
(i)~$\Rightarrow$~(ii) 
Let $a\in\fqn^*$ be such that $a\F_{q^2}\subseteq U_{s,\gamma}$ and let $\F_{q^2}=\langle 1, \eta\rangle_{\fq}$. Then there exist $u,v\in\fqk^*$ such that $a=v+v^{q^s}\gamma$ and $a\eta=u+u^{q^s}\gamma$.
Hence
\begin{equation}\label{e-eta}
    \frac{u+u^{q^s}\gamma}{v+v^{q^s}\gamma}=\eta\in\F_{q^2}\setminus\fq.
\end{equation}
(ii)~$\Rightarrow$~(iii) 
By assumption there exist $u,v\in\fqk^*$ such that \eqref{e-eta} is true. 
But this immediately implies $\gamma=(u-\eta v)(-u^{q^s}+\eta v^{q^s})^{-1}$, as desired.
One easily verifies that if $u/v\in\F_q$, then $\gamma$ is in $\F_{q^k}$, a contradiction.
\\
(iii)~$\Rightarrow$~(iv) First of all, (iii) implies that~$k$ is odd because otherwise $\eta$, and thus $\gamma$, is in $\F_{q^k}$, a contradiction.
Next, if $s$ is even then $\eta=\eta^{q^s}$ and thus $\gamma=-(u-\eta v)^{1-q^s}$.
Now \cref{L-theta}(c) implies $\N_{q^n/q^2}(\gamma)=\N_{q^n/q^2}(-1)\N_{q^n/q^2}((u-\eta v)^{1-q^s})=(-1)^k=-1$ since $k$ is odd.
If~$s$ is odd,  then $\gamma=-(u-\eta v)^{1-q^{s+k}}$ and again $\N_{q^n/q^2}(\gamma)=-1$.
\\
(iv)~$\Rightarrow$~(i) 
Let $\eta\in\F_{q^2}\setminus\F_q$ and note that $\F_{q^n}=\F_{q^k}(\eta)$.
\\
a) Let~$s$ be even. 
Since $\N_{q^n/q^2}(\gamma)= -1$,  \cref{L-theta}(b),(c) imply the existence of  $\xi \in \F_{q^{2k}}\setminus \F_{q^{k}}$ such that 
$\gamma =- \xi^{(1-q^s)}$. Moreover, there exist $u, v \in \F_{q^{k}}$ such that $\xi= u-\eta v$. This means that
\[
     \gamma=\frac{u-\eta v}{-u^{q^s}+\eta v^{q^s}}
\]
This implies \eqref{e-eta}. Set $a=v+v^{q^s}\gamma$, then $a\eta=u+u^{q^s}\gamma$.
But this means that $a\F_{q^2}=\subspace{1,\eta}\subseteq U_{s,\gamma}$, as desired.
\\
b) Let~$s$ be odd. Then $k+s$ is even and since $\N_{q^n/q^2}(\gamma)= -1$,  \cref{L-theta}(b),(c) imply the existence of  
$\xi \in \F_{q^{2k}}\setminus \F_{q^{k}}$ such that $\gamma =- \xi^{(1-q^{s+k})}$. 
As in the previous part we arrive at $\gamma=\frac{u-\eta v}{-u^{q^s}+\eta v^{q^s}}$ for some $u, v \in \F_{q^{k}}$, which in turn leads to 
$a\F_{q^2}\subseteq U_{s,\gamma}$ for  $a=v+v^{q^s}\gamma$.
\end{proof}

\section{Semilinear Isometries and Frobenius Isometries} 
In this short section we introduce and discuss the notions of  semilinear and Frobenius isometries.
Throughout we consider $\F_{q^n}$, where $q=p^h$ with~$p$ prime. 
Note that if $t$ is a divisor of~$hn$, then $\F_{q^n}$ is an $\F_{p^t}$-vector space of dimension $hn/t$.
If $t$ is a divisor of~$h$, then $\F_{p^t}$ is a subfield of~$\F_q$, and a $k$-dimensional $\F_q$-subspace of $\Fqn$ is an 
$kh/t$-dimensional $\F_{p^t}$-subspace of~$\Fqn$. 
Thus, for the Grassmannians we have the containment
\[
    \cG_q(k,n)\subseteq\cG_{p^t}(kh/t,nh/t).
\]

\begin{definition}
Let $t$ be a divisor of~$hn$. We define the groups
\[
   \GL_{hn/t}(p^t)=\{\phi:\F_{q^n}\longrightarrow\F_{q^n}\colon \phi\text{ is an $\F_{p^t}$-linear isomorphism}\}.
\]
\end{definition}

Note that $\GL_{hn/t}(p^t)\leq \GL_{hn}(p)$. Moreover, $\GL_n(q)$ is the group of $\F_q$-isomorphisms on $\F_{q^n}$, and $\F_{q^n}^*$ is isomorphic to 
the Singer group $S:=\GL_1(q^n)$ via $a\longmapsto m_a$, where $m_a:\Fqn\longrightarrow\Fqn,\ c\longmapsto ac$.
Furthermore, we define
\begin{equation}\label{e-Galpt}
     G_{p^t}=\Gal(\F_{q^n}\mid\F_{p^t})=\subspace{\sigma_{p^t}}, 
     \text{ where }\sigma_{p^t}:\Fqn\longrightarrow\Fqn,\ c\longmapsto c^{p^t}.
\end{equation}

\begin{remark}\label{R-GSSG}
Let $\F_{q^n}^*=\subspace{\omega}$. One easily checks that 
\[
    \sigma_{p^t}^\ell m_{\omega^s}=m_{\omega^{sp^{t\ell}}}\sigma_{p^t}^\ell\ \text{ and }\
    m_{\omega^s}\sigma_{p^t}^\ell=\sigma_{p^t}^\ell m_{\omega^{sp^{rn-t\ell}}}
     \ \text{ for all $\ell\in[hn/t]$ and $s\in[q^n-1]$.}
\]
Thus $G_{p^t}S=SG_{p^t}$, and this is a group of order $|G_{p^t}||S|=hn/t(q^n-1)$. In fact, it equals the normalizer $N_{\GL_{hn/t}(p^t)}(S)$ (see \cite[bottom of p.~98]{CoRe04} or \cite[top of p.~497]{Hes70}).
\end{remark}

Clearly,  the multiplicative group $\fqn^*$ does not depend on our choice of the ground field ($\F_q$ versus $\F_{p^t}$), and therefore 
the definition of the cyclic orbit code does not depend on this choice either.

\begin{definition}\label{D-Isometric}
Let $U,\,U'\in\cG_q(k,n)$ and $\cC=\Orb(U)$ and $\cC'=\Orb(U')$ be their cyclic orbit codes.
Let $t$ be a divisor of~$h$. 
\begin{enumerate}[label=(\alph*),leftmargin=1.5em,labelwidth=1.5em,labelsep=0em,align=left, topsep=-1.3ex]
\item $\cC$ and $\cC'$ are called {\bf $\F_{p^t}$-linearly isometric} if there exists $\phi\in\GL_{hn/t}(p^t)$ such that $\phi(\cC)=\cC'$.
If $t=1$, we call the codes {\bf semilinearly isometric}, and for $t=h$  they are called {\bf linearly isometric}.
\item $\cC$ and~$\cC'$ are called {\bf $\F_{p^t}$-Frobenius-isometric} if there exists $\phi\in G_{p^t}S$ such that  $\phi(\cC)=\cC'$.
\end{enumerate}
\end{definition}

\begin{remark}\label{R-Gpt}
Let~$\cC,\,\cC'$ be as in \cref{D-Isometric} and suppose $\cC$ and~$\cC'$ are $\F_{p^t}$-Frobenius-isometric.
Hence there exist $\ell\in[hn/t]$ and $\alpha\in\F_{q^n}^*$ such that  $\phi=\sigma_{p^t}^\ell m_{\alpha}$ satisfies $\phi(\cC)=\cC'$.
Since $m_{\alpha}(\cC)=\cC$, this implies $\sigma_{p^t}^\ell(\cC)=\cC'$.
Thus, we may replace the group $G_{p^t}S$ in \cref{D-Isometric}(b) by the Galois group~$G_{p^t}$.
\end{remark}

Note that the following definition is stronger than the one in  \cite[Def.~4.3]{Heideequiv}.

\begin{definition}\label{D-Generic}
A subspace $U\in\cG_q(k,n)$ is called \textbf{generic} if no cyclic shift of~$U$ is contained in a proper subfield of~$\F_{q^n}$.
\end{definition}

We have the following result.

\begin{theorem}\label{T-LinFrobIso}
Let $U,\,U'\in\cG_q(k,n)$ be generic subspaces and let $t$ be a divisor of~$h$.
Then $\Orb(U)$ and $\Orb(U')$ are $\F_{p^t}$-linearly isometric if and only if they are  $\F_{p^t}$-Frobenius-isometric.
\end{theorem}

\begin{proof}
Consider $U,\,U'$ as elements of the Grassmannian $\cG_{p^t}(kh/t,nh/t)$.
Then \cite[Thm.~6.2(a), Def.~3.5]{Heideequiv} applied to $q=p^t$ together with \cref{R-GSSG} tells us that being $\F_{p^t}$-linearly isometric is equivalent to 
being $\F_{p^t}$-Frobenius-isometric.
\end{proof}

The following example shows that semilinearly isometric codes are not necessarily linearly isometric.

\begin{example}
Let $q=4, \,n=4$, thus $\F_q=\F_4$ and $\F_{q^n}=\F_{4^4}$.
Choose a primitive element $\omega\in\F_{4^4}$. 
Since $(4^4-1)/(4-1)=85$, the element $a:=\omega^{85}$ is a primitive element of $\F_4$.
If we choose the element~$\omega$ with minimal polynomial $x^8+x^4+x^3+x^2+1$ over~$\F_2$, then its minimal polynomial over~$\F_4$ turns out to be $x^4+x^3+ax^2+ax+a$.
\\
In $\F_{4^4}=\F_{2^8}$ we consider the $\F_4$-subspaces 
\[
    \cU=\subspace{1,\omega}_{\F_4}=\subspace{1,\omega,a,a\omega}_{\F_2}\ \text{ and }\
    \cU'=\subspace{1,\omega^2}_{\F_4}=\subspace{1,\omega^2,a,a\omega^2}_{\F_2}.
\]
Note that both subspaces are generic and contain~$1$. 
Moreover, one can check that $\F_4$ is the maximal field of linearity for~$\cU$ and $\cU'$. 
Thus, $|\Orb(\cU)|=|\Orb(\cU')|=85$ for the cyclic orbit codes.
\\
Consider now the $2$-Frobenius $\sigma_2:\F_{4^4}\longrightarrow\F_{4^4},\ c\longmapsto c^2$, which 
is a semilinear automorphism on $\F_{4^4}$.
In other words, $\sigma_2\in\GL_8(2)$.  
Using that $a^2=1+a$, we observe that $\sigma_2(\cU)=\cU'$ and even more,
$\sigma_2(\omega^i\cU)=\omega^{2i}\cU'$ for all $i=0,\ldots,4^4-1$.
Hence $\sigma_2(\Orb(\cU))=\Orb(\cU')$, which means that these two orbits are semilinearly isometric.
By construction the codes are $\F_2$-Frobenius-isometric, which confirms the above theorem.
\\
Suppose $\Orb(\cU)$ and $\Orb(\cU')$ are linearly isometric (that is, via some $\psi\in\GL_4(4)$). 
Then \cref{T-LinFrobIso} implies that they are $\F_4$-Frobenius-isometric and \cref{R-Gpt} implies that there exists $\psi\in\Gal(\F_{4^4}\mid\F_4)$ such that $\psi(\Orb(\cU))=\Orb(\cU')$.
However, using for instance SageMath, one easily checks that $\psi(\cU)\not\in\Orb(\cU')$ for each $\psi\in\Gal(\F_{4^4}\mid\F_4)$.
Hence  semilinearly isometric codes are not necessarily linearly isometric.
\end{example}

\section{Automorphisms group of $\mathrm{Orb}(U_{s,\gamma})$}\label{S-AutGroups}

As in the previous section let $q=p^h$. 
We return to the orbit codes $\cO_{s,\gamma}=\Orb(U_{s,\gamma})$ introduced in \cref{D-Usgamma} and determine their
$\F_{p^t}$-automorphism groups for any divisor~$t$ of~$h$ as well as their $\F_p$-Frobenius automorphism groups.
Furthermore, we determine the number of distinct optimal (resp.\ quasi-optimal) full-length orbit codes $\cO_{s,\gamma}$ and the number of those
orbits where $U_{s,\gamma}$ contains a cyclic shift of $\F_{q^2}$.

We start with the following result from \cite{CPSZSidon}.
Recall the Galois groups $G_{p^t}$ from \eqref{e-Galpt} and the maps $m_\lambda:\Fqn\longrightarrow\Fqn,\ c\longmapsto \lambda c$.

\begin{theorem}[\mbox{\cite[Thm.~6.2]{CPSZSidon}}]\label{T-CPSZ23Theo}
Let $k\geq 2$ be a divisor of~$n$ and $U_1,\,U_2$ be $m$-dimensional $\F_q$-subspaces of $\F_{q^k}\times\F_{q^k}$.
Let $\gamma_i\in\Fqn\setminus\F_{q^k}$ and set $V_{U_i,\gamma_i}=\{u+\gamma_i u'\colon (u,u')\in U_i\}$ for $i=1,2$.
Suppose $V_{U_1,\gamma_1},V_{U_2,\gamma_2}$ are not contained in any cyclic shift $\alpha\F_{q^k}$, where $\alpha\in\Fqn$.
Let $\lambda\in\F_{q^n}^*$ and $\sigma\in G_{p}$.
The following are equivalent.
\begin{enumerate}[label=(\roman*),leftmargin=1.5em,labelwidth=1.5em,labelsep=0em,align=left, topsep=-1.3ex]
\item $m_{\lambda}\circ\sigma(V_{U_1,\gamma_1})=V_{U_2,\gamma_2}$.
\item There exists $A=\Smallfourmat{c}{d}{a}{b}\in\GL_2(\F_{q^k})$ such that 
        \[
             \gamma_2=\frac{a+b\sigma(\gamma_1)}{c+d\sigma(\gamma_1)},\quad \lambda=\frac{1}{c+d\sigma(\gamma_1)},\quad \text{and }
             \sigma(U_1)=U_2A,
        \]
        where $\sigma(U_1)=\{(\sigma(u),\sigma(u'))\colon (u,u')\in U_1\}$ and $U_2A=\{(u,u')A\colon (u,u')\in U_2\}$.
\end{enumerate} 
\end{theorem}

The previous result applies to every $\sigma\in\Gal(\Fqn\mid \F_p)$, thus to every $\sigma\in\Gal(\Fqn\mid \F_{p^t})$. 
By applying \cref{T-CPSZ23Theo} to the subspaces $U_{s,\gamma}$ we get the following result.
Recall the notation~$\theta_q$ from \cref{L-theta}.

\begin{proposition}
\label{prop:7.27.3}
Let $k>2$ and $n=2k$. For $i=1,2$ let $s_i\in[k]$ be such that $\gcd(s_i,k)=1$. 
Moreover, let $\gamma_i\in\Fqn\setminus\F_{q^k}$ for $i=1,2$ and set
$U_{s_i,\gamma_i}=\{a+\gamma_ia^{q^{s_i}}\colon a\in\F_{q^k}\}$.
Let $\sigma\in G_p$.
Then the following are equivalent.
\begin{enumerate}[label=(\roman*),leftmargin=1.5em,labelwidth=1.5em,labelsep=0em,align=left, topsep=-1.3ex]
\item $\sigma(\Orb(U_{s_1,\gamma_1}))=\Orb(U_{s_2,\gamma_2})$.
\item Either 
      \[
           s_1\neq s_2\ \text{ and }\ s_1+s_2\equiv0\ (\mmod k)\ \text{ and }\ \sigma(\gamma_1)\gamma_2\in\theta_{q}(\F_{q^k})
      \]
      or
      \[
            s_1=s_2\ \text{ and }\ s_1+s_2\not\equiv0\ (\mmod k)\ \text{ and }\ \gamma_2/\sigma(\gamma_1)\in\theta_{q}(\F_{q^k}).
      \]
\end{enumerate}
\end{proposition}

\begin{proof}
For $i=1,2$ set $U_{s_i}=\{(x,x^{q^{s_i}})\colon x\in\F_{q^k}\}$.  Then $U_{s_i,\gamma_i}=V_{U_{s_i},\gamma_i}$ as defined in \cref{T-CPSZ23Theo}.
Moreover,
\begin{align}
  \sigma(\Orb(U_{s_1,\gamma_1}))=\Orb(U_{s_2,\gamma_2})&\Longleftrightarrow 
   m_{\lambda}\circ\sigma(U_{s_1,\gamma_1})=U_{s_2,\gamma_2}\text{ for some }{\lambda}\in\F_{q^n}^* \nonumber\\
   &\Longleftrightarrow  m_{\lambda}\circ\sigma(V_{U_{s_1},\gamma_1})=V_{U_{s_2},\gamma_2}\text{ for some }\lambda\in\F_{q^n}^*.\label{e-lamsig}
\end{align}
In order to apply \cref{T-CPSZ23Theo}, we have to show that the spaces $U_{s_i,\gamma_i}$ are not contained in any cyclic shift of $\F_{q^k}$. 
Assume to the contrary that
$U_{s_i,\gamma_i}= \alpha\F_{q^k}$ for some $\alpha\in\fqn$ (note that $\dim_{\F_q}(U_{s_i,\gamma_i})=k$). 
Then $1+\gamma_i=\alpha \bar{a}$ for some $\bar{a}\in\F_{q^k}^*$.
Moreover, for any $a\in\fqk$ there exists $b_a\in\fqk$ such that $a+\gamma_i a^{q^{s_i}}=\alpha b_a$.
This leads to $\bar{a}a+\gamma_i\bar{a}a^{q^{s_i}}=b_a+b_a\gamma_i$. 
Thus $a\bar{a}=\bar{a}a^{q^{s_i}}$ for any $a\in\fqk$. This implies $\bar{a}=0$, a contradiction.
Now we may apply \cref{T-CPSZ23Theo} to \eqref{e-lamsig} and arrive at 
\begin{equation}\label{e-EquivOrbits}
    \sigma(\Orb(U_{s_1,\gamma_1}))=\Orb(U_{s_2,\gamma_2})\Longleftrightarrow 
    \left\{\begin{array}{l}
        \text{there exists }A=\Smallfourmat{c}{d}{a}{b}\in\GL_2(\F_{q^k})\ \text{ such that}\\[.7ex]
        \gamma_2=\frac{a+b\sigma(\gamma_1)}{c+d\sigma(\gamma_1)}\ \text{ and }\sigma(U_{s_1})=U_{s_2}A.\end{array}\right.
\end{equation}
Since~$\sigma$ induces a ring isomorphism on $\F_{q^k}$, we have
$\sigma(U_{s_1})=\{(\sigma(y),\sigma(y^{q^{s_1}}))\colon y\in\F_{q^k}\}=\{(y,y^{q^{s_1}})\colon y\in\F_{q^k}\}$.
Using $|U_1|=q^k=|U_2|=|U_2A|$, we obtain the chain of equivalences
\begin{align*}
   \sigma(U_{s_1})=U_{s_2}A&\Longleftrightarrow 
     \{(y,y^{q^{s_1}})\colon y\in\F_{q^k}\}=\{(cx+ax^{q^{s_2}},dx+bx^{q^{s_2}})\colon x\in\F_{q^k}\}\\
     &\Longleftrightarrow\left\{\begin{array}{l}\text{for all $x\in\F_{q^k}$ there exists $y\in\F_{q^k}$ such that}\\
            cx+ax^{q^{s_2}}=y\ \text{ and }\  dx+bx^{q^{s_2}}=y^{q^{s_1}}\end{array}\right.\\
     &\Longleftrightarrow\left\{\begin{array}{l}\text{the polynomial }\hat{P}=dX-c^{q^{s_1}}X^{q^{s_1}}+bX^{q^{s_2}}-a^{q^{s_1}}X^{q^{s_1+s_2}}\in\F_{q^k}[X]\\
         \text{has every $x\in\F_{q^k}$ as a root} \end{array}\right.\\ 
     &\Longleftrightarrow\left\{\begin{array}{l}\text{the polynomial }P=dX-c^{q^{s_1}}X^{q^{s_1}}+bX^{q^{s_2}}-a^{q^{s_1}}X^{q^{(s_1+s_2)(\text{mod } k)}}\\
         \text{is the zero polynomial,}\end{array}\right.
\end{align*}
where the last step follows from the fact that $P$ has degree less than $q^k$ and at least $q^k$ roots.
In order to identify the coefficients of~$P$, we need to distinguish several cases. Set $\ell=(s_1+s_2)\ (\mmod k)$. 
\\
\underline{Case 1:} $\ell=0$ and $s_1\neq s_2$. Then $P=(d-a^{q^{s_1}})X-c^{q^{s_1}}X^{q^{s_1}}+bX^{q^{s_2}}$. 
Hence $P=0$ if and only if $c=b=0$ and $d=a^{q^{s_1}}$. 
Using \eqref{e-EquivOrbits} this implies $\gamma_2=a/(a^{q^{s_1}}\sigma(\gamma_1))$, hence $\sigma(\gamma_1)\gamma_2=a^{1-q^{s_1}}$.
Thus $\sigma(\gamma_1)\gamma_2\in\theta_q(\F_{q^k})$ thanks to \cref{L-theta}(b).
This proves (i)~$\Rightarrow$~(ii) for this case. For the converse, let $\sigma(\gamma_1)\gamma_2\in\theta_q(\F_{q^k})$.
Again by \cref{L-theta} we have $\sigma(\gamma_1)\gamma_2=a^{1-q^{s_1}}$ for some $a\in\F_{q^k}^*$. 
Then \eqref{e-EquivOrbits} is satisfied with the matrix $A=\Smallfourmat{0}{a^{q^{s_1}}}{a}{0}$.
\\
\underline{Case 2:} $\ell=0$ and $s_1=s_2$. In this case $2s_1\equiv0\ (\mmod k)$, contradicting the assumptions $k>2$ and $\gcd(s_1,k)=1$.
Hence this case does not arise.
\\
\underline{Case 3:} $\ell\neq0$ and $s_1=s_2$. Then $P=dX+(b-c^{q^{s_1}})X^{q^{s_1}}-a^{q^{s_1}}X^\ell$ and
$P=0$ if and only if $a=d=0$ and $b=c^{q^{s_1}}$. 
Using \eqref{e-EquivOrbits} this implies $\gamma_2=c^{q^{s_1}}\sigma(\gamma_1)/c$, hence $\gamma_2/\sigma(\gamma_1)=c^{q^{s_1}-1}$.
Thus $\gamma_2/\sigma(\gamma_1)\in\theta_q(\F_{q^k})$ thanks to \cref{L-theta}(b).
Now we may complete the reasoning as in Case~1.
\\
\underline{Case 4:} $\ell\neq0$ and $s_1\neq s_2$. 
In this case, all four terms of~$P$ involve different powers of~$X$ and thus 
$P=0$ if and only if $a=b=c=d=0$, which contradicts that $\det(A)\neq0$.
Hence this case does not arise.
\end{proof}

\begin{remark}
Let $k>2$. On the set  $\cA=\{(s,\gamma)\colon s\in[k],\,\gcd(s,k)=1,\, \gamma\in\Fqn\setminus\F_{q^k}\}$
define the relation
\[
   (s_1,\gamma_1)\sim(s_2,\gamma_2)\Longleftrightarrow 
   \left\{\begin{array}{l}
    s_1\neq s_2,\ \, s_1+s_2\equiv0\ (\mmod k) \text{ and } \sigma(\gamma_1)\gamma_2\in\theta_q(\F_{q^k})^*\text{ for some }\sigma\in G_p,\\
    \text{or}\\
    s_1=s_2,\ \, s_1+s_2\not\equiv0\ (\mmod k) \text{ and } \gamma_2/\sigma(\gamma_1)\in\theta_q(\F_{q^k})^*\text{ for some }\sigma\in G_p.
    \end{array}\right.
\]
Using that~$\theta_q(\F_{q^k})^*$ is a group and invariant under $G_p$, 
one easily verifies that~$\sim$ is an equivalence relation on~$\cA$.
The equivalence classes are thus in bijection with the $\F_p$-Frobenius-isometry classes of the orbit codes  $\cO_{s,\gamma}$.
Moreover, if we restrict ourselves to~$\sigma=\id_{\F_{q^n}}$, then we obtain an equivalence relation~$\approx$, whose equivalence classes are in bijection to
the  orbit codes  $\cO_{s,\gamma}$.
\end{remark}

Now we obtain the following description of the automorphism groups of the orbits codes $\cO_{s,\gamma}$.
Recall the notation from \eqref{e-Galpt} and the fact that $\GL_{hn/t}(p^t)\leq \GL_{hn}(p)$.

\begin{proposition}\label{P-AutGroup}
Let $\cO_{s,\gamma}$ be as in \cref{D-Usgamma}. 
For any divisor $t$ of~$h$ denote its automorphism group  in $\GL_{hn/t}(p^t)$ by $\Aut_{p^t}(\cO_{s,\gamma})$. Then
\[
     \Aut_{p^t}(\cO_{s,\gamma})=\Big\{\sigma_{p^t}^im_a\colon a\in\F_{q^n}^*,\,i\in[hn/t]\text{ such that }\gamma^{p^{it}-1}\in\theta_q(\F_{q^k})\Big\},
\]
where $\theta_{q}$ is as in \cref{L-theta}.
\end{proposition}

\begin{proof}
As shown in the proof of  \cref{prop:7.27.3}, $U_{s,\gamma}$ does not coincide with a cyclic shift of $\F_{q^k}$. 
Since $\dim_{\F_q} U_{s,\gamma}=k=n/2$, we conclude that $U_{s,\gamma}$ is generic in the sense of Definition \ref{D-Generic}. 
\\
Now, we turn to the stated identity.
Let $\tau\in\Aut_{p^t}(\cO_{s,\gamma})$. 
Consider~$U_{s,\gamma}$ as a subspace in the Grassmannian $\cG_{p^t}(nh/t,kh/t)$. 
Then \cite[Thm.~4.4]{Heideequiv} implies that $\Aut_{p^t}(\cO_{s,\gamma})\leq N_{GL_{hn/t}(p^t)}(S)=G_{p^t}S$.
Thus, $\tau=\sigma_{p^t}^i m_a$ for some $i\in[hn/t]$ and $a\in\F_{q^n}^*$. 
Clearly, $m_a\in\Aut_{p^t}(\cO_{s,\gamma})$ and thus $\sigma_{p^t}^i\in\Aut_{p^t}(\cO_{s,\gamma})$.
Hence $\sigma_{p^t}^i(\Orb(U_{s,\gamma}))=\Orb(U_{s,\gamma})$ and \cref{prop:7.27.3} implies 
$\gamma^{p^{it}-1}=\gamma^{p^{it}}/\gamma\in\theta_{q}(\F_{q^k})$.
The converse containment follows in the same way.
\end{proof}

We now engage in some counting and start with determining the number of distinct orbits $\cO_{s,\gamma}$.
For the rest of this section we fix the following.

\begin{notation}\label{Nota}
Let  $k>2$ and $n=2k$. 
Set $\cO=\{\cO_{s,\gamma}\colon s\in[k],\,\gcd(s,k)=1,\,\gamma\in\Fqn\setminus\F_{q^k}\}$, which
 is the set of distinct orbits of the spaces~$U_{s,\gamma}$.
Furthermore, fix a primitive element~$\omega$ of~$\Fqn$. 
\end{notation}

\begin{proposition}\label{P-DistinctOrbits}
Set 
\begin{align*}
   \cL&=\{\ell\colon \ell=1,\ldots,(q^k+1)(q-1),\,(q^k+1)\nmid\ell\},\\[.5ex]
   \bar{\cA}&=\{(s,\ell)\colon 1\leq s\leq\lfloor k/2\rfloor,\,\gcd(s,k)=1,\,\ell\in\cL\}.
\end{align*}
Then  $\alpha:\bar{\cA}\longrightarrow\cO,\ (s,\ell)\longmapsto\cO_{s,\omega^\ell}$
is a bijection and thus $|\cO|=|\bar{\cA}|=\phi(k)q^k(q-1)/2$, where $\phi$ is the Euler-$\phi$-function.
\end{proposition}

\begin{proof}
Applying \cref{prop:7.27.3} to $\sigma=\id_{\F_{q^n}}$ we conclude that for any two orbits $\cO_{s_1,\gamma_1}$ and $\cO_{s_2,\gamma_2}$ we have
\begin{equation}\label{e-SameOrbits}
    \cO_{s_1,\gamma_1}=\cO_{s_2,\gamma_2}\Longleftrightarrow 
    \left\{\begin{array}{l}s_2=s_1\text{ and }\gamma_2\in\gamma_1\theta_q(\F_{q^k}^*)\\ \text{or }\\
         s_2=k-s_1\text{ and } \gamma_2\in\gamma_1^{-1}\theta_q(\F_{q^k}^*).\end{array}\right.
\end{equation}
Using that $n=2k$ we have $\F_{q^k}^*=\subspace{\omega^{q^k+1}}$ and $\theta_q(\F_{q^k}^*)=\subspace{\omega^{(q^k+1)(q-1)}}$.
A complete set of distinct coset representatives of $\theta_q(\F_{q^k}^*)$ is thus given by $\omega^{\ell},\,\ell=1,\ldots,(q^k+1)(q-1)$. 
Clearly, if one representative is in $\Fqn\setminus\F_{q^k}$, then this is true for every representative. 
Thus, for parametrizing the codes $\cO_{s,\gamma}$ we only have to consider the set of representatives $\omega^\ell$, where $\ell\in\cL$. 
It is clear from \eqref{e-SameOrbits} and the definition of~$\bar{\cA}$ that~$\alpha$ is injective. 
For surjectivity consider an orbit $\cO_{s_1,\gamma_1}$. 
If $1\leq s_1\leq\lfloor k/2\rfloor$, then choose the unique $\ell\in\cL$ such that 
$\gamma_1\in\omega^{\ell}\theta_q(\F_{q^k}^*)$. Hence $\cO_{s_1,\gamma_1}=\cO_{s_1,\omega^\ell}=\alpha(s_1,\ell)$.
If $s_1>\lfloor k/2\rfloor$, then the unique $\ell\in\cL$ such that 
$\gamma_1^{-1}\in\omega^{\ell}\theta_q(\F_{q^k}^*)$ leads to $\cO_{s_1,\gamma_1}=\cO_{k-s_1,\omega^\ell}$.
For the cardinality note that $|\{s\colon 1\leq s\leq\lfloor k/2\rfloor,\,\gcd(s,k)=1\}|=\phi(k)/2$ and 
$|\cL|=(q^k+1)(q-1)-(q-1)=q^k(q-1)$.
\end{proof}

Recall from \cref{P-OptQuasiOpt} that $\cO_{s,\gamma}$ is either an optimal or a quasi-optimal full-length orbit code. 
We can now count these two cases.

\begin{proposition}\label{P-CountOpt}
We have 
\[
   \big|\{\cO_{s,\gamma}\colon \cO_{s,\gamma}\text{ is quasi-optimal}\}\big|
   =\left\{\begin{array}{cl}\frac{\phi(k)}{2}q^k,&\text{if $q$ is even},\\[.5ex]  \frac{\phi(k)}{2}(q^k-1),&\text{if $q$ is odd.}\end{array}\right.
\]
The number of optimal codes in~$\cO$ is therefore $\phi(k)q^k(q-2)/2$ if~$q$ is even and $\phi(k)(q^k(q-2)+1)/2$ if~$q$ is odd.
\end{proposition}

\begin{proof}
Consider the set~$\cL$ from \cref{P-DistinctOrbits} and set 
$\cL_{\text{qo}}=\{\ell\in\cL\colon\N_{q^n/q}(\omega^\ell)=1\}$.
From \cref{P-OptQuasiOpt} we know that $\cO_{s,\omega^\ell}$ is quasi-optimal $\Longleftrightarrow$ $\ell\in\cL_{\text{qo}}$.
Clearly, $\N_{q^n/q}(\omega^\ell)=1\Longleftrightarrow\omega^{\ell\frac{q^{2k}-1}{q-1}}=1\Longleftrightarrow(q-1)\mid\ell$ and
thus $\cL_{\text{qo}}=\{\ell\colon \ell\leq(q^k+1)(q-1),\, (q-1)\mid \ell,\,(q^k+1)\nmid\ell\}$.
Hence
\[
  |\cL_{\text{qo}}|=q^k+1-\frac{(q^k+1)(q-1)}{\lcm(q^k+1,q-1)}=q^k+1-\gcd(q^k+1,q-1).
\]
Since $q^k+1=(q-1)\sum_{i=0}^{k-1}q^i+2$, we obtain $|\cL_{\text{qo}}|=q^k$ if~$q$ is even and $|\cL_{\text{qo}}|=q^k-1$ if~$q$ is odd.
Taking the number of values for~$s$ into account, we arrive at the desired result.
\end{proof}

Recall from \cref{thm:thm4.1heide} that the shape of the distance distribution of $\Orb(U_{s,\gamma})$ depends on whether 
$U_{s,\gamma}$ 
contains a cyclic shift of the subfield $\F_{q^2}$.
We now determine the number of such orbit codes.
Note that the fraction $(q^k+1)/(q+1)$ appearing below is indeed an integer for~$k$ odd.

\begin{proposition}\label{P-CycShift}
We have 
\[
   \big|\{\cO_{s,\gamma}\colon U_{s,\gamma}\text{ contains a cyclic shift of $\F_{q^2}$}\}\big|
   =\left\{\begin{array}{cl}\frac{\phi(k)}{2}\Big(\frac{q^k+1}{q+1}-1\Big),&\text{if $k$ is odd},\\[1ex] 0,&\text{if $k$ is even.}\end{array}\right.
\] 
\end{proposition}

\begin{proof}
By \cref{T-ContainCycShift} $U_{s,\gamma}$ contains a cyclic shift of $\F_{q^2}$ if and only if $k$ is odd and $\N_{q^{2k}/q^2}(\gamma)=-1$.
Thus from now on let~$k$ be odd.
With $\gamma=\omega^\ell$ we compute
\begin{align*}
   \N_{q^{2k}/q^2}(\omega^\ell)=-1&\Longleftrightarrow \omega^{\ell\frac{q^{2k}-1}{q^2-1}}=-1\\
   &\Longleftrightarrow 
    \left\{\begin{array}{cl}(q^2-1)\mid 2\ell\text{ and }(q^2-1)\nmid\ell,&\text{if $q$ odd},\\[.5ex]
              (q^2-1)\mid \ell,&\text{if $q$ even.}\end{array}\right.
\end{align*}
a) Let~$q$ be even. Set $\cL'=\{\ell\in\cL\colon (q^2-1)\mid\ell\}$. 
Using that~$\cL'$ consists of those positive integers upper bounded by $(q^k+1)(q-1)$ that are multiples of $q^2-1$, but not of 
$\lcm(q^k+1,q^2-1)$, we have
\[
  |\cL'|=\frac{(q^k+1)(q-1)}{q^2-1}-\frac{(q^k+1)(q-1)}{\lcm(q^k+1,q^2-1)}.
\]
Since $(q^k+1)/(q+1)=(q-1)\sum_{i=1}^{(k-1)/2}q^{k-2i}+1$, we have $\gcd((q^k+1)/(q+1),q-1)=1$ and thus 
$\gcd(q^k+1,q^2-1)=q+1$. As a consequence, $\lcm(q^k+1,q^2-1)=(q^k+1)(q^2-1)/(q+1)=(q^k+1)(q-1)$, and we arrive at the desired result.
\\
b) Let~$q$ be odd. Set $\cL'=\{\ell\in\cL\colon (q^2-1)/2\mid\ell\text{ and }(q^2-1)\nmid\ell\}$ and $N:=(q^k+1)(q-1)$. 
Then $\cL'=\{\ell\colon 1\leq\ell\leq N,\,(q^2-1)/2\mid\ell,\,(q^2-1)\nmid\ell,\,(q^k+1)\nmid\ell\}$.
There are $2N/(q^2-1)$ such multiples of~$(q^2-1)/2$. 
Of those, $N/(q^2-1)$ are multiples of~$q^2-1$. 
Now we arrive at 
\begin{align*}
  |\cL'|&=\frac{2N}{q^2-1}-\frac{N}{q^2-1}-\frac{N}{\lcm((q^2-1)/2,q^k+1)}+\frac{N}{\lcm(q^2-1,q^k+1)}\\
        &=\frac{N}{q^2-1}-\frac{2N}{(q^k+1)(q-1)}+\frac{N}{(q^k+1)(q-1)}=\frac{q^k+1}{q+1}-1,
\end{align*}
where the second step follows from $\lcm((q^2-1)/2,q^k+1)=(q^k+1)(q-1)/2$.
\end{proof}

We now turn to the $\F_{p}$-Frobenius automorphism group of an orbit $\cO_{s,\omega^\ell}$.
The automorphism groups $\Aut_{p}(\cO_{s,\omega^\ell})$ have been determined in Proposition \ref{P-AutGroup}.
\cref{R-GSSG} implies that $\Aut_{p}(\cO_{s,\omega^\ell})=(\Aut_{p}(\cO_{s,\omega^\ell})\cap G_p)S$, 
where $S=\GL_1(q^n)$.
We now describe the group $\Aut_{p}(\cO_{s,\omega^\ell})\cap G_p$, that is, the stabilizer of  
$\cO_{s,\omega^\ell}$ under the action of the Galois group~$G_p$.

\begin{definition}\label{D-FrobAuto}
Let $\cO_{s,\omega^\ell}\in\cO$. We set $\cF_{s,\omega^\ell}:=\Aut_{p}(\cO_{s,\omega^\ell})\cap G_p$, and call this group the 
\textbf{Frobenius-automorphism group} of $\cO_{s,\omega^\ell}$.
Moreover, we denote the orbit of $\cO_{s,\omega^\ell}$ under the action of $G_p$ by 
$\Orb_{G_p}(s,\omega^\ell)=\{\sigma_p^i(\cO_{s,\omega^\ell})\colon i\in[hn]\}$ and call it the \textbf{Frobenius-orbit} of $\cO_{s,\omega^\ell}$.
Note that $|\Orb_{G_p}(s,\omega^\ell)|=hn/|\cF_{s,\omega^\ell}|$ (by the Orbit-Stabilizer Theorem).
\end{definition}

Let us rephrase the criterion given in  \cref{P-AutGroup}.
For $\gamma=\omega^\ell$ and $i\in\{0,\ldots,hn-1\}$ we have
$\gamma^{p^i-1}\in\theta_q(\F_{q^k}^*)\Longleftrightarrow \omega^{\ell(p^i-1)}\in\subspace{\omega^{(q^k+1)(q-1)}}$, and 
this in turn is equivalent to 
\begin{equation}\label{e-congruence}
   \ell(p^i-1)\equiv j(q^k+1)(q-1)\ \mmod(q^n-1) \text{ for some }j\in\Z.
\end{equation}
Thus we will first study this congruence.

\begin{proposition}\label{P-MinSol}
Fix $i\in\{1,\ldots,hn-1\}$ and set 
\[
    \hat{\ell}_i=\frac{\lcm(p^i-1,\,(q^k+1)(q-1))}{p^i-1}.
\]
Then $\hat{\ell}_i\leq \frac{(q^k+1)(q-1)}{p-1}$ and $\{\ell\in\Z\colon \ell\text{ satisfies \eqref{e-congruence}}\}=\hat{\ell}_i\Z$.
\end{proposition}

\begin{proof}
The first statement follows from the identity $\lcm(a,b)=ab/\gcd(a,b))$ for any $a,b\in\Z$.
Next note that the solution set to \eqref{e-congruence} is an ideal in~$\Z$. 
Moreover,~$\hat{\ell}_i$ satisfies \eqref{e-congruence} because $\hat{\ell}_i(p^i-1)=\lcm(p^i-1,\,(q^k+1)(q-1))=j(q^k+1)(q-1)$
for some $j\in\Z$.
This shows that $\hat{\ell}_i\Z$ is contained in the solution set of \eqref{e-congruence}.
Finally, an arbitrary $\ell\in\Z$ satisfies \eqref{e-congruence} if and only if $(q^n-1)\mid (\ell(p^i-1)-j(q^k+1)(q-1))$ for some $j\in\Z$. 
Using that $q^n-1=(q^k+1)(q-1)\frac{q^k-1}{q-1}$, this implies that $(q^k+1)(q-1)\mid \ell(p^i-1)$.
The smallest positive~$\ell$ satisfying this last condition is~$\hat{\ell}_i$. 
This concludes the proof.
\end{proof}

Now we are ready to provide a numerical description of the Frobenius-automorphism groups.

\begin{theorem}\label{T-AutGroupellhat}
Set $\cI=\{i\in[hn]\colon (q^k+1)\nmid \hat{\ell}_i\}$.
Let $(s,\ell)\in\bar{\cA}$ and consider the orbit~$\cO_{s,\omega^\ell}$. Then
\[
    \cF_{s,\omega^\ell}=\left\{\begin{array}{cl}
        \{\id_{\Fqn}\},&\text{ if $\hat{\ell}_i\nmid\ell$ for all }i\in\cI,\\[.5ex]
         \subspace{\sigma_p^\iota},\text{ where }\iota=\min\{i\in\cI\colon \hat{\ell}_i\text{ divides } \ell\},&\text{ if $\hat{\ell}_i\mid\ell$ for some }i\in\cI.\end{array}\right.
\]
As a consequence, the Frobenius-automorphism groups of the orbits~$\cO_{s,\omega^\ell}$ 
are exactly the groups $\subspace{\sigma_p^i}$, $i\in\hat{\cI}$, where $\hat{\cI}:=\{i\in\cI\colon i=\min\{t\colon \hat{\ell}_t=\hat{\ell}_i\}\}$.
Even more, $\subspace{\sigma_p^i},\,i\in\hat{\cI},$ arises as the Frobenius-automorphism group for exactly
$\frac{\phi(k)}{2}\big(\alpha_i-\sum_{j\in\hat{\cI}\setminus\{i\},\, j\mid i}\alpha_j\big)$ orbits~$\cO_{s,\omega^\ell}$, where 
\[
     \alpha_i=\frac{(q^k+1)(q-1)}{\hat{\ell}_i}-\frac{(q^k+1)(q-1)}{\lcm(\hat{\ell}_i,q^k+1)}.
\]
\end{theorem}

It is not hard to see that the fractions involved in the expression for $\alpha_i$ are indeed integers:
we have $\hat{\ell}_i(p^i-1)=b(q^k+1)(q-1)$ for some $b\in\Z$ such that $\gcd(\hat{\ell}_i,b)=1$. 
Hence $\hat{\ell}_i$ divides $(q^k+1)(q-1)$.
For the second fraction use that $\lcm(\hat{\ell}_i,q^k+1)=\hat{\ell}_i(q^k+1)/\gcd(\hat{\ell}_i,q^k+1)$ and a Bezout identity for 
the greatest common divisor.

\begin{proof}
Fix $(s,\ell)\in\bar{\cA}$. Then \eqref{e-congruence} and  \cref{P-MinSol} implies for $i\geq1$
\[
   \sigma_p^i\in\Aut(\cO_{s,\omega^\ell})\Longleftrightarrow  \omega^{\ell(p^i-1)}\in\subspace{\omega^{(q^k+1)(q-1)}}
   \Longleftrightarrow \hat{\ell}_i\mid\ell.
\]
From the definition of the set~$\bar{\cA}$, we know that $q^k+1$ does not divide~$\ell$.
Hence the same is true for every $\hat{\ell}_i$ dividing~$\ell$.
If there is no such $\hat{\ell}_i$, we thus have $\cF_{s,\omega^\ell}=\{\id_{\F_{q^n}}\}$. 
Otherwise, $\cF_{s,\omega^\ell}=\{\id_{\F_{q^n}}\}\cup\{\sigma_p^i\colon i\in\cI,\,\hat{\ell}_i\text{ divides }\ell\}$.
Since~$G_p$ is cyclic, the subgroup $\cF_{s,\omega^\ell}$ is thus generated by $\sigma_p^\iota$, where~$\iota$ is the smallest
such value~$i$.
This proves the first two statements.
\\
As for the counting, note that for any $i\in\hat{\cI}$
\[
   \subspace{\sigma_p^i}\leq\cF_{s,\omega^\ell}\Longleftrightarrow 
        \hat{\ell}_i\mid\ell\text{ and }(q^k+1)\nmid\ell\text{ and }\ell<(q^k+1)(q-1).
\]
This is satisfied for~$\alpha_i$ values for~$\ell$.
Moreover, 
\[
    \subspace{\sigma_p^i}=\cF_{s,\omega^\ell}\Longleftrightarrow \subspace{\sigma_p^i}\leq\cF_{s,\omega^\ell} \text{ and }
    \subspace{\sigma_p^j}\not\leq\cF_{s,\omega^\ell} \text{ for all proper divisors $j\in\hat{\cI}$ of~$i$.} 
\]
This is satisfied by $\alpha_i-\sum_{j\in\hat{\cI}\setminus\{i\},\, j\mid i}\alpha_j$ many values for~$\ell$.
Realizing that the factor $\phi(k)/2$ accounts for the possible values for~$s$ in the index set~$\bar{\cA}$, we arrive at the final result.
\end{proof}

\begin{example}\label{E-NumExa}
\begin{enumerate}[label=(\alph*),leftmargin=1.5em,labelwidth=1.5em,labelsep=0em,align=left, topsep=-1.3ex]
\item Let $p=q=3$ (thus $h=1$) and $k=3$ (thus $n=6$). 
      Thanks to Proposition \ref{P-DistinctOrbits} we have $\phi(k)q^k(q-1)/2=54$ distinct orbit codes, and they are parametrized by $(1,\omega^\ell)$, where 
      $\ell\in\cL$, which is defined in the same proposition.
      As shown in \cref{P-OptQuasiOpt}, those satisfying $N_{q^n/q}(\omega^\ell)\neq 1$ form optimal full-length orbit codes, while all others are quasi-optimal ones.
      By \cref{P-CountOpt} we have 26 quasi-optimal codes and 28 optimal codes, and by \cref{P-CycShift} 6 codes consist of subspaces containing a cyclic shift of~$\F_{q^2}$ (this is Case~(4) below).
      Furthermore, $\hat{\ell}_i$ as in \cref{P-MinSol} is given by $(\hat{\ell}_1,\ldots,\hat{\ell}_{n-1})=(28,\,7,\,28,\,7,\,28)$
      and since $q^k+1=28$ we have $\cI=\{2,4\}$ and $\hat{\cI}=\{2\}$.
      Thus, the only non-trivial Frobenius-automorphism group is $\subspace{\sigma_3^2}$ (which has order~$3$). 
      It arises for $\alpha_1=6$ orbit codes~$\cO_{1,\omega^\ell}$ and they form~$3$ Frobenius-orbits of length~$hn/3=2$. 
      All other $48$ codes $\cO_{1,\omega^\ell}$ have trivial Frobenius-automorphism group and they thus form 
      $48/hn=8$ Frobenius-orbits of length~$hn=6$.
      In this small case, one can also compute the orbit codes~$\cO_{1,\omega^\ell}$ directly and obtains:
      \begin{enumerate}[label=(\arabic*),leftmargin=1.5em,labelwidth=1.5em,labelsep=0em,align=left, topsep=0ex]
      \item 4 Frobenius orbits of length 6 consisting of optimal full-length orbit codes,
      \item  2 Frobenius orbits of length 2 consisting of optimal full-length orbit codes,
      \item  3 Frobenius orbits of length 6 consisting of quasi-optimal full-length orbit codes
               whose subspaces do not contain a cyclic shift of $\F_{q^2}$,
      \item 1 Frobenius orbit of length 6 consisting of quasi-optimal full-length orbit codes
               whose subspaces contain a cyclic shift of $\F_{q^2}$,
      \item 1 Frobenius orbit of length 2 consisting of quasi-optimal full-length orbit codes
              whose subspaces do not contain a cyclic shift of $\F_{q^2}$.
      \end{enumerate}
      The codes in~(4) are of type as in \cref{thm:thm4.1heide}(1)
      and have intersection distribution with $\lambda_2=3$.
      The codes in~(3) and~(5) are of type as in \cref{thm:thm4.1heide}(2) and all of them satisfy $\lambda_2=12$.
      All of this agrees with \cref{cor:fUconsequencequasi-optimal}.
      It also implies that in all cases $\lambda_1\neq0$; see also \cref{R-k3Case}.
\item Let $q=p=2$ and $k=5$ (thus $n=10$).  
        In this case we have $\phi(k)q^k(q-1)/2=64$ distinct orbit codes, and they all form quasi-optimal full-length codes 
        thanks to \cref{P-OptQuasiOpt}.
        By \cref{P-CycShift}, 20 codes consist of subspaces containing a cyclic shift of~$\F_{q^2}$ (this is Case~(2) below).
        We obtain
       \begin{enumerate}[label=(\arabic*),leftmargin=1.5em,labelwidth=1.5em,labelsep=0em,align=left, topsep=-1ex]
        \item 2 Frobenius orbits of length 10; both are of type \cref{thm:thm4.1heide}(1) and satisfy $\lambda_2=134$.
        \item 4 Frobenius orbits of length 10; they are of type \cref{thm:thm4.1heide}(2) and satisfy $\lambda_2=150$.
        \item  2 Frobenius orbits of length 2; both are of type \cref{thm:thm4.1heide}(2) and satisfy $\lambda_2=150$.
        \end{enumerate}

\item Let $p=3,\,h=3$, thus $q=27$, and let $k=4$ (thus $n=8$). In this case we have $\phi(k)q^k(q-1)/2=13,817,466$ distinct orbit codes, 
      and they are again parametrized by $(1,\omega^\ell),\,\ell\in\cL$.
      Thanks to \cref{T-ContainCycShift} none of them consists of subspaces containing a cyclic shift of~$\F_{q^2}$.
      Furthermore, by \cref{P-CountOpt} $531,440$ are quasi-optimal.
      Using for instance SageMath, one easily finds $\cI=\{2,4,6,\ldots,22\}$ and $\hat{\cI}=\{2,6,8\}$ (that is, every $\hat{\ell}_i,\,i\in\cI$, is a multiple of one of 
      $\hat{\ell}_2,\,\hat{\ell}_6,\,\hat{\ell}_8$).
      Thus, since $hn=24$, all non-trivial Frobenius-automorphism groups have order~$12,\,4$, or~$3$, and hence all Frobenius-orbits have length~$2,\,6,\,8$, or $24$.
      Using \cref{T-AutGroupellhat} one finds 
      \[
         \begin{array}{ll}
          \text{1 Frobenius-orbit of size~$2$},&\quad\text{4 Frobenius-orbits of size~$6$},\\
          \text{20 Frobenius-orbits of size~$8$},&\quad\text{575,720 Frobenius-orbits of size~$hn=24$}.
          \end{array}
      \]
      This accounts for all $13,817,466$ elements. 
      By \cref{P-OptQuasiOpt}, the size of each orbit code $\cO_{1,\omega^\ell}$ equals $(q^n-1)/(q-1)>10^{10}$, which prohibits any direct computation of their distance distribution.
\end{enumerate}
\end{example}

\section{Open Problems}

(1) For small parameters $q$ and~$n=2k$ we determined all orbits codes $\cO_{s,\gamma}$, and their distance distribution appears to be quite rigid:
it only depends on $n,\,q$, and the fact whether $U_{s,\gamma}$ contains a cyclic shift of~$\F_{q^2}$
(in other words, the parameters~$r$ in \cref{thm:thm4.1heide} only depends on these data).
We have to leave it to future research to prove or disprove this conjecture.
\\
(2) \cref{thm:existencequasi-optimal} and \cref{thm:classificationcase3dim} guarantee the existence of quasi-optimal cyclic orbit codes when $n$ is even or $k=3$. We leave as an open problem the existence of such codes in the remaining cases.
\\
(3) In this paper, we have mainly focused on the case of one-orbit codes. Can the results of this paper be used to construct larger quasi-optimal cyclic codes consisting of more than one orbit?

\section*{Acknowledgements}

The research was partially supported by the project ``COMPACT" of the University of Campania ``Luigi Vanvitelli'' and  by the Italian National Group for Algebraic and Geometric Structures and their Applications (GNSAGA - INdAM).


Chiara Castello, Olga Polverino and Ferdinando Zullo\\
Dipartimento di Matematica e Fisica,\\
Universit\`a degli Studi della Campania ``Luigi Vanvitelli'',\\
I--\,81100 Caserta, Italy\\
{{\em \{chiara.castello,olga.polverino,ferdinando.zullo\}@unicampania.it}}

\medskip

Heide Gluesing-Luerssen\\
Department of Mathematics,\\ 
University of Kentucky,\\ 
Lexington KY 40506--0027, USA\\ 
{\em heide.gl@uky.edu}

\end{document}